\documentclass[10pt,conference,onecolumn]{IEEEtran}
\usepackage{amsmath}
\usepackage{latexsym}
\usepackage{amssymb}

\usepackage{ifpdf}
\ifpdf
\usepackage[pdftex]{graphicx}
\usepackage{epstopdf}
\else
\usepackage{graphicx}
\fi

\newtheorem{theorem}{Theorem}
\newtheorem{lemma}{Lemma}
\newtheorem{proposition}{Proposition}

\title{Co-evolution of Content Popularity and Delivery in Mobile P2P Networks}
\begin{document}
\author{
\IEEEauthorblockN{Srinivasan Venkatramanan and Anurag Kumar}
\IEEEauthorblockA{Department of Electrical Communication Engineering\\
Indian Institute of Science\\
Bangalore 560012, India\\
Email: vsrini,anurag@ece.iisc.ernet.in}
}
\maketitle
\begin{abstract}
Mobile P2P technology provides a scalable approach to content delivery to a large number of users on their mobile devices. In this work, we study the dissemination of a \emph{single} content (e.g., an item of news, a song or a video clip) among a population of mobile nodes. Each node in the population is either a \emph{destination} (interested in the content) or a potential \emph{relay} (not yet interested in the content). There is an interest evolution process by which nodes not yet interested in the content (i.e., relays) can become interested (i.e., become destinations) on learning about the popularity of the content (i.e., the number of already interested nodes). In our work, the interest in the content evolves under the \emph{linear threshold model}. The content is copied between nodes when they make random contact. For this we employ a controlled epidemic spread model. We model the joint evolution of the copying process and the interest
evolution process, and derive the joint fluid limit ordinary differential equations. We then study the selection of the parameters under the content provider's control, for the optimization of various objective functions that aim at maximizing content popularity and efficient content delivery.
\end{abstract}

\section{Introduction}
Peer-Peer (P2P) architectures help relieve file servers from excessive load by enabling clients to communicate among themselves and exchange content, hence providing a scalable approach to content delivery for a large audience. With the proliferation of smart phones and high speed mobile Internet, there is increasing download of media content directly into mobile devices. Mobile P2P technology aims to leverage this trend and several possible architectures \cite{diaz-etal07mobilep2p-survey} have been proposed. In this paper we are concerned with peer-to-peer spread of the content under the delay tolerant networking (DTN) paradigm.

While studying content distribution, it is essential to understand the evolution of demand for the content. For modeling demand evolution, we can adopt the point of view that the existing demand for an item influences others also to get interested in it. There has been considerable interest in such models in the area of viral marketing research. For example, online social networks (Facebook, Twitter, etc.) provide a platform for the users to exchange their views about a particular media content with members of their social circle and end up \emph{influencing} each other. Content creators interested in increasing the popularity of the content use several data mining techniques to identify optimal viral marketing strategies \cite{domingos02mining-viral-marketing} and to identify the most influential users in a social network \cite{kempe-etal03max-spread-infl}. In this paper, we are also concerned with spread of interest in the item of media content, and for this we adopt an influence spread model introduced earlier in the context of viral marketing \cite{kempe-etal03max-spread-infl}. 

Evidently, content providers need to adopt their content dissemination strategies according to the prevailing interest in the particular content. Though there have been several works that discuss optimal strategies for content distribution in mobile ad hoc networks in the DTN setting \cite{khouzani-etal10patch-dissemination},\cite{singh-etal11dtn-multi-destination} and in mobile P2P \cite{repantis-kalogeraki04data-dissem-mp2p}, they do not address the issue of jointly modeling the spread of interest in the content and the content itself. In our work, we aim to model the joint evolution of demand and spread of the content in a mobile P2P setting. 

A recent work \cite{shakkottai-etal10demand-aware-content-spread} addresses a related problem of demand-aware distribution by employing fluid models for the viral spread of demand and aims at obtaining hybrid P2P and client-server architectures that can meet the demand.  While \cite{shakkottai-etal10demand-aware-content-spread} assumes that the demand spread is uncontrolled and optimizes for content delivery, in this work, we separately discuss optimizing the parameters of the demand evolution process (to increase the popularity the content), as well as the content delivery process (to efficiently meet the content demand). 

\hspace{-0.3cm}\textit{Contributions:}
Our main contributions are:
\begin{itemize}
 \item Developing ordinary differential equation (o.d.e.\ ) models for the co-evolution of popularity of the content and its spread via controlled epidemic copying in a mobile P2P environment
 \item Using these o.d.e. approximations to provide insights into the choice of parameters for the content provider in order to optimize various system performance objectives
\end{itemize}

In Section~\ref{sec:math-model} we introduce the combined system model and the related notations. In Section~\ref{sec:HILT-fluid} we analyse a threshold based model for evolution of interest and derive its fluid limit for general threshold distributions. We also discuss the effect of threshold distribution on the model, and show that SIR epidemic model \cite{kermack-mckendrick27SIR-epidemics} is a special case of our model. In Section~\ref{sec:HILT-SI-fluid} we derive fluid limits for evolution of spread on top of the interest evolution model. We finally provide numerical results for some practical optimizations (Section~\ref{sec:numerical}) associated with the evolution of interest and the joint-evolution of interest and spread of the content.

\section{The System Model}
\label{sec:math-model}
We consider a scenario where the users with their mobile devices, modeled as a population of mobile \emph{nodes}, meet each other according to a random contact model. We assume that the nodes are homogeneous and we wish to model the spread of a \emph{single} content among these nodes. As motivated in the Introduction, we aim to model the joint spread of popularity of the content and the content itself. Pertaining to the content, each node has a current state represented by two bits: the \emph{want} bit indicates whether the node is interested in the content and the \emph{have} bit indicates whether the node has the content. When the \emph{want} bit of a node is $1$, we call the node, a \emph{destination} (i.e., interested in the content), else we call it a \emph{relay} (i.e., not yet interested in the content). 

There is a central server that keeps track of the \emph{want} state of all the nodes. The total number of destinations (nodes interested in the content) is a measure of popularity of the content. In order to promote the content, this measure of popularity must be made known to all the nodes. We assume that there is a low bitrate control channel that (i) is used by the nodes to inform the central server about their interest in the content, and (ii) is used by the central server to broadcast, at regular intervals, the fraction of nodes interested in the content. Thus, for the purpose of spread of content popularity, the network is fully connected. 

Relay nodes, on receiving popularity broadcasts from the central server, might get converted to destinations according to an \emph{influence spread model}. We model this influence process (evolution of interest in the content) using the Homogeneous Influence Linear Threshold (HILT) model, a special case of the Linear Threshold model (LT) introduced in Kempe et al.\ \cite{kempe-etal03max-spread-infl}. For the content copying process, we model the random contacts between pairs of nodes as independent Poisson point processes (a model also used in the context of mobile P2P in \cite{queija-etal08scaling-laws-p2p}), with the copying (when pairs of nodes meet) being controlled probabilistically, in a manner similar to the Susceptible-Infective (SI) epidemic model \cite{allen94dicrete-time-SIR}.

\subsection{Modeling Interest Evolution: The HILT Model}
\label{sec:HILT-model}
In this section we introduce the HILT model used to model the evolution of interest in the content. In the original LT model \cite{kempe-etal03max-spread-infl}, nodes are part of a weighted directed graph $\mathcal{G}=(\mathcal{N},\mathcal{E})$, where $\mathcal{E} \subseteq \mathcal{N} \times \mathcal{N}$. With each ${i,j} \in \mathcal{E}$, there is associated a weight $w_{i,j}$ which gives a measure of \emph{influence} of node $i$ on node $j$, normalized such that the total weight into any node is at most 1, i.e., $\sum_{i} w_{i,j} \leq 1$. The Homogeneous Influence Linear Threshold (HILT model) is a special case of the LT model where the network graph is complete and all nodes are homogeneous. Hence, we have a mesh network on the population $\mathcal{N}$ containing $N=|\mathcal{N}|$ nodes with each edge carrying the same influence weight $\gamma_N=\frac{\Gamma}{N-1}$ and $\Gamma \leq 1$ (see Figure~\ref{fig:hilt}). 
\begin{figure}
\centering
\includegraphics[scale=0.3]{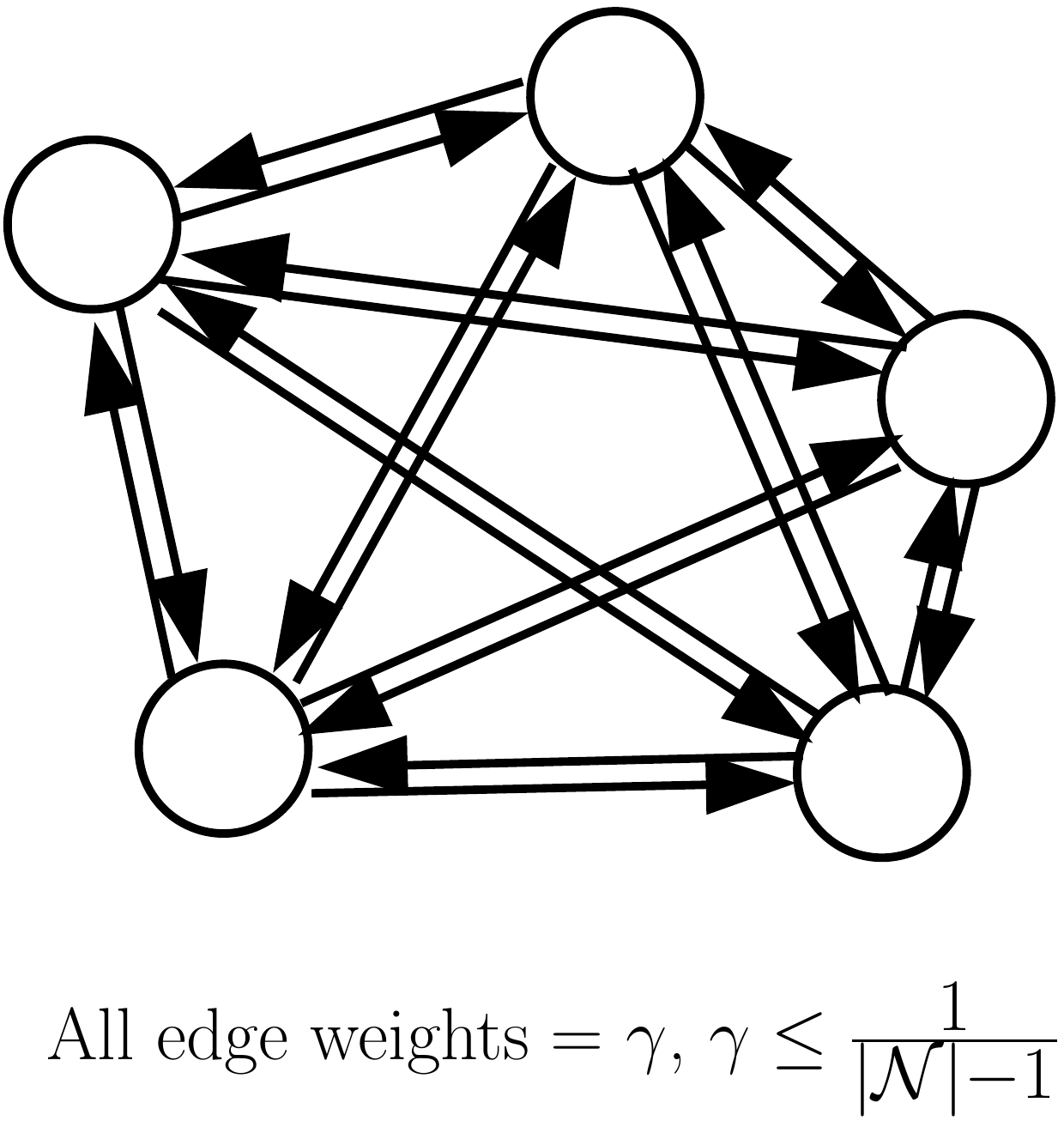}
\vspace{-0.3cm}
\caption{The HILT Network Model}
\label{fig:hilt}
\vspace{-0.3cm}
\end{figure}

The evolution of interest in the content is modeled by the following influence process evolving in discrete time. Each node $j \in \mathcal{N}$ independently chooses a random threshold $\Theta_{j}$ from a given distribution $F$ \emph{at the beginning}. We begin with an initial set of destinations $\mathcal{A}(0)$.  In the HILT model, the net influence of a set of destinations on any relay is $\gamma_N$ times the size of the destination set. The number of destinations (currently interested nodes) is broadcast to all the nodes by the central server after each period, and a relay gets converted into a destination once the net influence exceeds its chosen threshold. We also assume the \emph{progressive case}, i.e., conversion into destinations is an irreversible process.  In other words, a relay $j \notin \mathcal{A}(k-1)$ gets influenced in step $k$ if, 
\begin{eqnarray}
\label{eqn:HILT-activation}
 \gamma_N | \mathcal{A}(k-1) |  \geq \Theta_{j} 
\end{eqnarray}
The influence is modeled as building up cumulatively over time. The initial set $\mathcal{A}(0)$ is viewed as being \emph{infectious} and results in some of the relay nodes ``tipping'' over their thresholds and getting converted to destinations in the first period. A node $j$ that remains a relay  has the level of cumulative influence on it raised to $\gamma_N |\mathcal{A}(0)| (< \Theta_j)$. The nodes in $\mathcal{A}(0)$ are now viewed as being non-infectious, and the newly infected nodes are denoted by $\mathcal{D}(1)$, with $\mathcal{A}(1) = \mathcal{A}(0) + \mathcal{D}(1)$; see Figure~\ref{fig:spread_influence}. Thus, at the end of each period the population will contain three types of nodes: $\mathcal{A}(k)$, the set of non-infectious destinations, $\mathcal{D}(k)\subseteq \mathcal{A}(k)$, the set of newly infected destinations in that period (infectious for the next period) and the set of relays denoted by $\mathcal{S}(k)$. It is clear that the activation process stops at a random step $T$ when there are no more infectious destinations, i.e., $\mathcal{D}(T) = \emptyset$, and a \emph{terminal set} $\mathcal{A}(T)$ is reached. 

\begin{figure}
\centering
\includegraphics[scale=0.45]{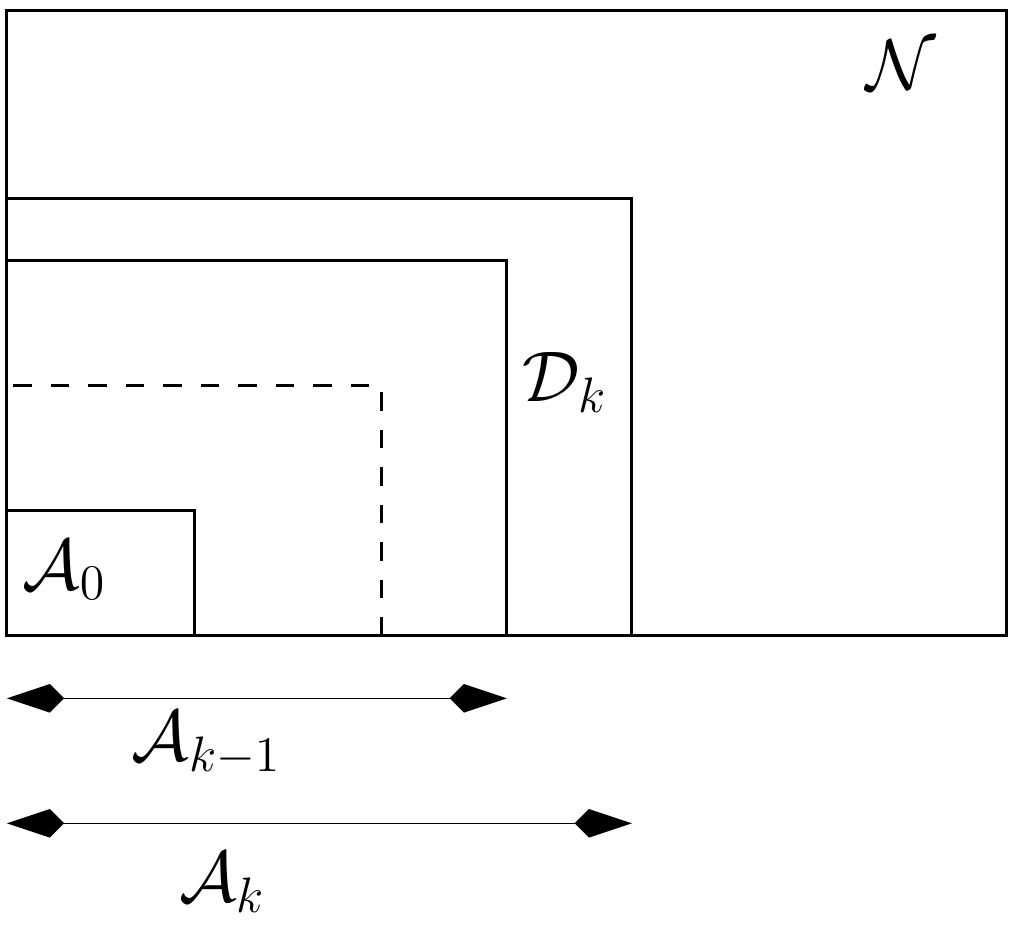}
\vspace{-0.3cm}
\caption{Spread of Influence in the HILT Model}
\label{fig:spread_influence}
\vspace{-0.5cm}
\end{figure}

\subsection{Modeling Content Copying: SI Model}
\label{sec:si-model}
In order to model the content delivery process, we further classify the nodes depending on whether they have the content (i.e., based on the \emph{have} bit). Let $\mathcal{X}(k) \subseteq \mathcal{A}(k)$ denote the set of destinations that have the content, and $\mathcal{Y}(k) \subseteq \mathcal{S}(k)$ denote the set of relays that have the content.

For the evolution of $(\mathcal{X}(k), \mathcal{Y}(k))$, we model a content copying process based on the Susceptible-Infective (SI model). Between the discrete time steps, pairs of nodes meet at the points of a Poisson process with rate $\lambda$, and whenever a node that has the content meets a node that does not, content transfer takes place in a probabilistic manner. At such a meeting, we distinguish between the node that does not have the content, being a destination or a relay, by having different copy probabilities $\alpha_N = \frac{\alpha}{N}$ and $\sigma_N = \frac{\sigma}{N}$, respectively. The scaling with respect to $N$ is to ensure that the fluid limit infection rates would be $\alpha$ and $\sigma$ respectively \cite{allen94dicrete-time-SIR}. 

\subsection{Co-evolution of Interest and Spread: The HILT-SI Model}
\label{sec:hilt-si-model}
In the combined model, underlying the content delivery process (the SI model), the influence process (the HILT model) converts relay nodes into destinations. Thus, in our setup, the fraction of destinations is time-varying (unlike \cite{singh-etal11dtn-multi-destination}). Also, the content spread is dependent on the interest evolution but not vice versa. Figure~\ref{fig:hilt-si} shows the possible transitions between the four sets of nodes. 
\begin{figure}
\centering
\includegraphics[scale=0.3]{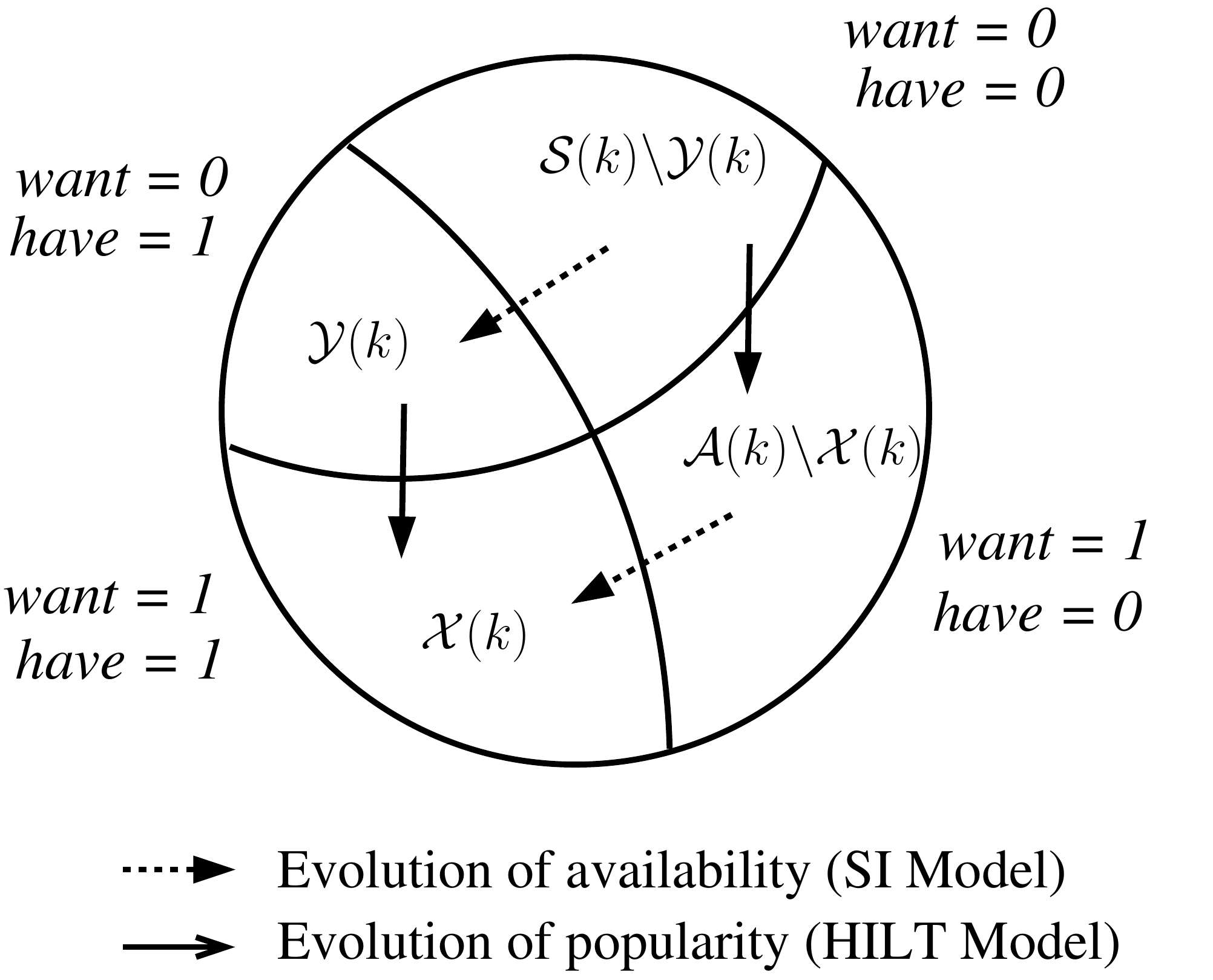} 
\vspace{-0.3cm}
\caption{Joint evolution of content popularity and spread in the HILT-SI Model.}
\label{fig:hilt-si}
\vspace{-0.5cm}
\end{figure}
An interesting feature of this model is the importance of copying to a relay node. As a content provider, we might be interested in delivering only to the destinations (interested in the content). But, there are two advantages of copying to a relay. First, copy to a relay promotes the further spread of the content even to destinations; this is the aspect explored in a controlled Markov process setting in \cite{singh-etal11dtn-multi-destination}. Second, the relay we copy to now might later get influenced (by the HILT model) to become a destination, which is a unique feature of the HILT-SI model. 

\section{Evolution of Interest in the Content}
\label{sec:HILT-fluid}
In this section and in Section V, we will use the convergence of a scaled discrete time Markov chain (DTMC) to a deterministic limit described by an ordinary differential equation (o.d.e.) (Kurtz \cite{kurtz70ode-markov-jump-processes}, and also Darling \cite{darling02limits-purejump-markov}, \cite{darling-norris08ode-markov-chains}), to develop deterministic (or, so called, fluid) approximations to the  HILT model for interest evolution, and the HILT-SI model for co-evolution of interest and content spread. 

Consider the HILT model on $N$ nodes, and with edge weights $\gamma_N = \frac{\Gamma}{N-1}$. Recall that $\Theta_i$ is the influence threshold of User~$i$, $1 \leq i \leq N$. The $\Theta_i, 1 \leq i \leq N,$ are non-negative, i.i.d.\ random variables chosen according to the cumulative distribution function $F(\cdot)$, a continuous distribution with density function $f(\cdot)$.  We see that $\Gamma$ and the threshold distribution $F$ together define the HILT model. They can be seen as modeling the susceptance level of the population to being influenced by the popularity of the content under consideration. 

Note that $\mathcal{A}(k), k \geq 1$ represents the set of destinations and $\mathcal{D}(k)$ the set of \emph{newly added} destinations (infectious) at the end of period $k$ with $\mathcal{D}(0) = \mathcal{A}(0)$. Define $\mathcal{B}(k) = \mathcal{A}(k-1)$  ,  the set of destination nodes up to period $k-1$ with $\mathcal{B}(0) = \emptyset$. Thus, for $k \geq 1$, $\mathcal{B}(k) = \cup_{0 \leq i \leq k-1} \mathcal{D}_i$. We will be working with sets $\mathcal{B}(k)$ and $\mathcal{D}(k)$ to derive the fluid limits of the HILT process. Also, since the nodes are homogeneous in the HILT model, it suffices to record the sizes of the respective sets, and not on the exact membership of the sets themselves.  Let $A(k)$, $D(k)$ and $B(k)$ be the sizes of the sets $\mathcal{A}(k)$, $\mathcal{D}(k)$ and $\mathcal{B}(k)$ respectively. 
 
\subsection{O.D.E.\ Model for Interest Evolution}
\label{sec:ODE_interest}
We can show that the original HILT process $(B(k), D(k))$ is a Markov chain (see Appendix~\ref{app:HILT-markov}). In order to obtain an approximating o.d.e., we work with an appropriately scaled Markov process $(B^{N}(t), D^{N}(t))$, which evolves on a time scale $N$ times faster than that of the original system. We can visualize this process as evolving over ``minislots'' of duration $1/N$, whereas the original process evolves at the epochs $k=0,1,2,\cdots$. The minislots are indexed by $t=0,1,2,\cdots$. Since this new process runs on a faster time scale, we slow down its dynamics by utilizing an approach taken in \cite{benaim-leboudec08mean-field-models}. In the present context, the scaling approach can be interpreted as follows. In each minislot, each infectious destination in $D^{N}(t)$ is permitted to influence the relays with probability $\frac{1}{N}$ and its influence is deferred with probability $1-\frac{1}{N}$. In the former case, it contributes its influence of $\frac{\Gamma}{N-1}$ and then moves to the set $B^{N}(t+1)$, otherwise it stays in the $D^{N}(t+1)$ set. In the original process, the influence of all the newly infected destinations will be taken into account by the central server while announcing the popularity level of the content in the next time step, whereas, in the scaled process, only those infectious destinations that choose to use their influence will be considered. Define by $C^N(t) \subseteq D^N(t)$ the set of infectious destinations who use their influence at time $t$. Then, 
\[ C^{N}(t) = \frac{D^{N}(t)}{N} + Z^{N}_{b}(t+1)\]
\[ B^{N}(t+1) = B^{N}(t) + C^{N}(t) \]
and by applying Equation~\ref{eqn:HILT-activation} to the relay nodes $j \in \mathcal{N} \backslash \mathcal{A}(k)$, 
\begin{eqnarray*}
\lefteqn{D^{N}(t+1) = D^{N}(t) - C^{N}(t)}  \\
\hspace{-2cm} &+& \mathbb{E} \bigg[ \frac{F(\gamma_N(B^{N}(t) + C^N(t)))- F(\gamma_N B^{N}(t))}{1-F(\gamma_N B^{N}(t))} \bigg] \times \\
& & \bigg(N-B^{N}(t) -D^{N}(t)\bigg) + Z^{N}_{d}(t+1)
\end{eqnarray*}
where $Z^{N}_b(t+1)$ and $Z^{N}_d(t+1)$ are zero mean random variables conditioned on the history of the process $(B^{N}(0), D^{N}(0), \cdots B^{N}(k), D^{N}(k))$ representing the noise terms, and the expectation in the expression for $D^{N}(t+1)$ is with respect to $C^N(t)$. Define $\Tilde{B}^{N}(t) = \frac{B^{N}(t)}{N}$ and similarly $\Tilde{C}^{N}(t)$ and $\Tilde{D}^{N}(t)$ to denote the fraction of users of each type (destination and relay). $(\Tilde{B}^{N}(t), \Tilde{D}^{N}(t))$, $t=0,1,2,\cdots$ is a DTMC on the state space $[0, \frac{1}{N}, \frac{2}{N}, \cdots 1] \times [0, \frac{1}{N}, \cdots 1]$ with $\Tilde{B}^{N}(t)+\Tilde{D}^{N}(t)\leq 1$. 

Then we can state the following theorem on the convergence of the scaled process to their deterministic limits.
\begin{theorem}
\label{thm:kurtz-theorem-hilt}
Given the interest evolution Markov process $( \Tilde{B}^{N}(t), \Tilde{D}^{N}(t) )$, with bounded $\dot{f}(\cdot)$ for the threshold distribution,  we have for each $T > 0$ and each $\epsilon > 0$,
\begin{eqnarray*}
\mathbb{P} \bigg( \sup_{0 \leq u \leq T} \big| \big| \big( \Tilde{B}^{N}( \lfloor Nu \rfloor ),\Tilde{D}^{N}( \lfloor Nu \rfloor ) \big) - \big( b(u),d(u) \big) \big| \big| > \epsilon \bigg)\\
\stackrel{N\rightarrow \infty}{\rightarrow} 0
\end{eqnarray*}
where $(b(u),d(u))$ is the unique solution to the o.d.e.,
\[ \dot{b} = d\]
\[ \dot{d} = \frac{f(\Gamma b)\Gamma d }{1- F(\Gamma b)} (1-b-d) - d \]
with initial conditions $(b(0)=0,d(0)=a(0))$.
\end{theorem}

\begin{proof}
This is essentially an instance of Kurtz's Theorem \cite{kurtz70ode-markov-jump-processes}; see also \cite{darling02limits-purejump-markov}, \cite{darling-norris08ode-markov-chains}. In Appendix~\ref{app:kurtz} we derive equivalent conditions for verifying Kurtz's theorem. Refer Appendix~\ref{app:hilt-fluidlimit} for the steps involved in deriving the fluid limit. A detailed proof verifying the necessary conditions for Kurtz's theorem for the HILT case is provided in Appendix~\ref{app:kurtz-applied-hilt}.
\end{proof}
\textit{Remark:}
We know that the hazard function\cite{cox61renewal-theory} for the cdf $F(x)$ is given by $h_F(x) = \frac{f(x)}{1-F(x)}$ where $f(x)$ is the corresponding probability density function. Hence the o.d.e. becomes,
\begin{eqnarray} 
\label{eqn:ODE_F_1}
\dot{b} &=& d  \\
\label{eqn:ODE_F_2}
 \dot{d} &=& h_F(\Gamma b) \Gamma d  (1-b-d) - d 
\end{eqnarray}
\subsection{Accuracy of the O.D.E.\ Approximation}
Consider the HILT process of interest evolution, under the special case of uniform distribution of influence thresholds, i.e., $\Theta_i \sim U[0,1]$. The hazard function corresponding to uniform distribution is given by \[h_F(x) = \frac{1}{1-x}\] The corresponding o.d.e.s then become:
\begin{eqnarray}
\label{eqn:ODE_unif_1}
\dot{b} &=& d\\
\label{eqn:ODE_unif_2}
\dot{d} &=& \frac{\Gamma d}{1-\Gamma b}  (1-b-d) - d 
\end{eqnarray}
Figure ~\ref{plot:hilt-convergence} shows the convergence of the scaled HILT process $(\Tilde{B}^N(\lfloor Nt \rfloor), \Tilde{D}^N(\lfloor Nt \rfloor))$ to the solutions of the above o.d.e. $(b(t), d(t))$ for increasing values of $N$ (50,100,500,1000). $\Gamma$ and $d_0 $ were chosen to be $0.9$ and $0.2$ respectively. We see thatfor  $N= 500, 1000$ the o.d.e. approximates the scaled HILT Markov chain very well.

The original HILT process of interest evolution $(B(k), D(k))$ is then compared with the solutions of the o.d.e. $(b(t), d(t))$. The results are shown in Figure ~\ref{plot:kurtz-errorbar}, where multiple sample paths of the original discrete time HILT process are superimposed on the o.d.e. solutions. We find that the o.d.e. solution approximates the original process really well, and hence permits using the fluid limit approximation for sufficiently large $N$.
\begin{figure}
\centerline{\includegraphics[scale=0.3]{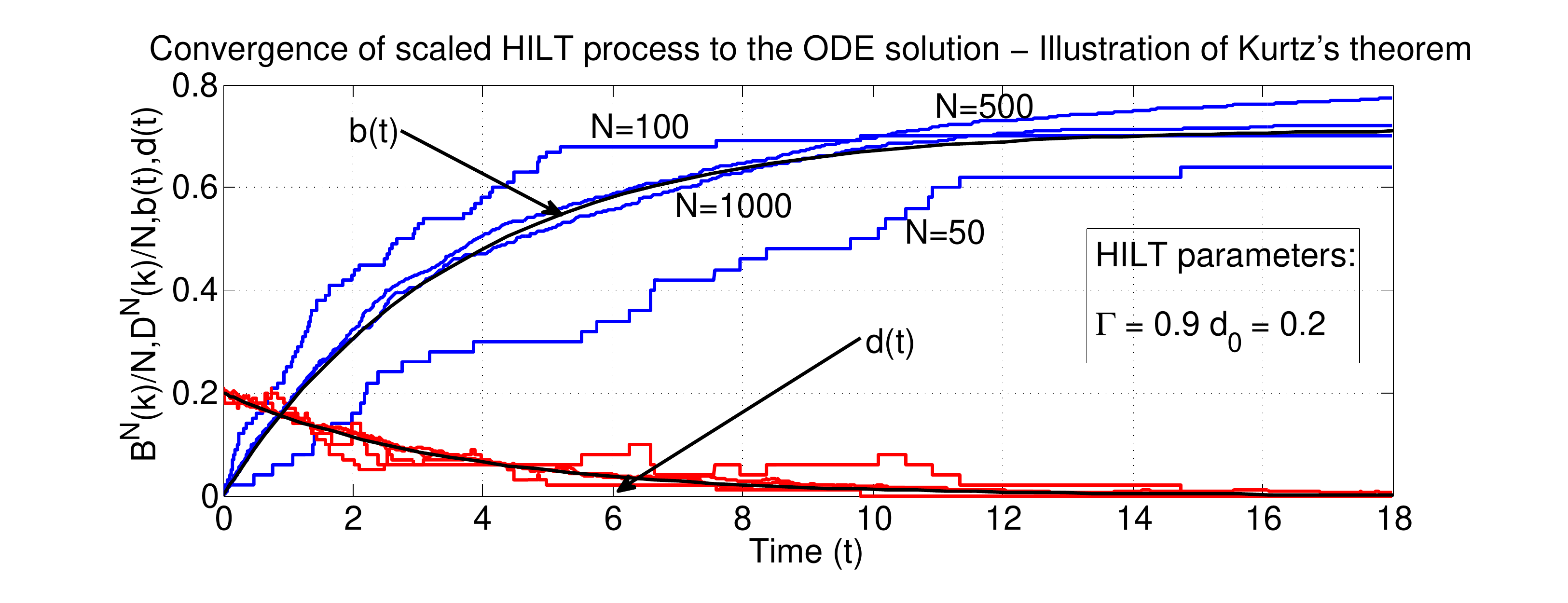}}
\vspace{-0.3cm}
\caption{Convergence of the scaled HILT Markov chain $(\Tilde{B}^{N}(k), \Tilde{D}^{N}(k))$ to the o.d.e. limit $(b(t), d(t))$. The o.d.e. solution $(b(t),d(t))$ plotted along with sample paths of the scaled process $(\Tilde{B}^{N}(k), \Tilde{D}^{N}(k))$ for $N=50,100,500,1000$.}
\label{plot:hilt-convergence}
\vspace{-0.5cm}
\end{figure}

\begin{figure}
\centerline{\includegraphics[scale=0.3]{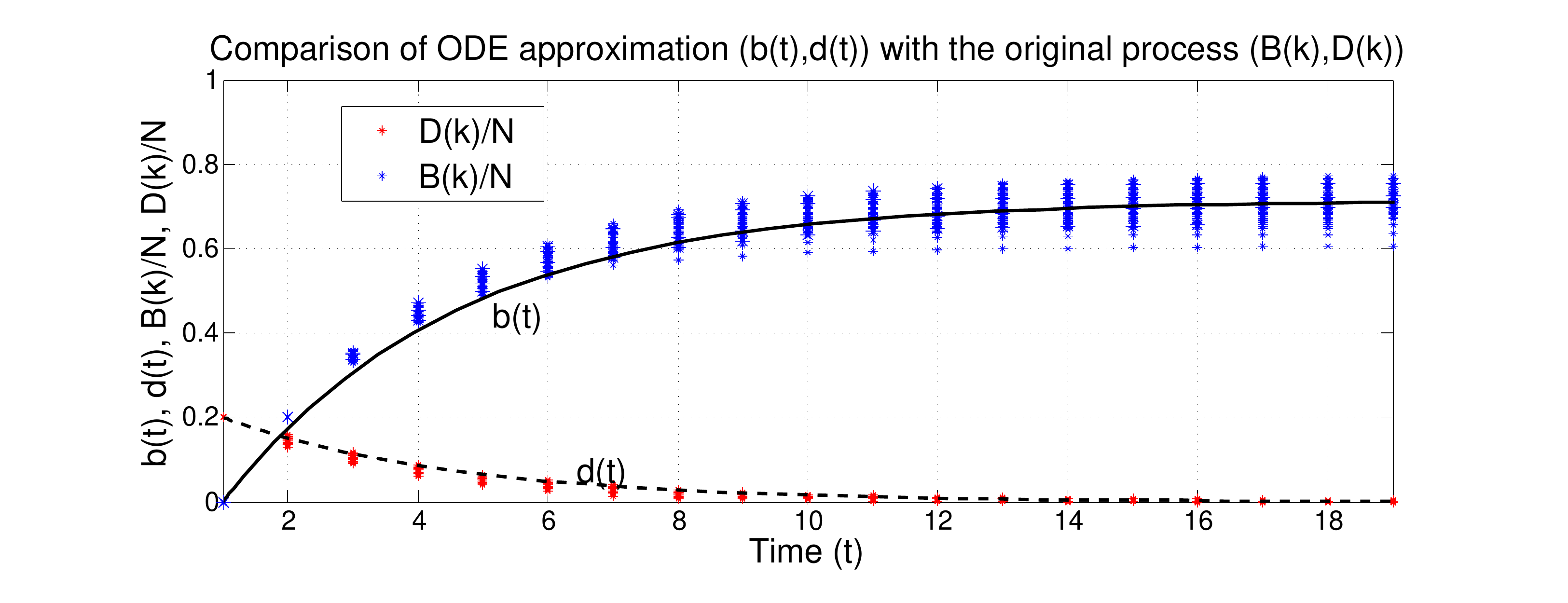}}
\vspace{-0.3cm}
\caption{Multiple sample paths of the unscaled HILT process $(B(k),D(k))$ for $N=3000$ is plotted along with the o.d.e. solution $(b(t),d(t))$ for $\Gamma = 0.9$ and $d_0=0.2$}
\label{plot:kurtz-errorbar}
\vspace{-0.5cm}
\end{figure}

\subsection{Effect of the Threshold distribution}
\label{sec:HILT-distribution}
In the HILT model, while $\Gamma$ is indicative of the total level of influence each individual can receive from the others, the threshold distribution $F(\cdot)$ captures the variation among the individuals' susceptance levels for getting interested the content. An empirical analysis on the effect of threshold distributions on collective behavior is available in \cite{granovetter78threshold-models}. Having established the o.d.e. limit and studied its ability to approximated the evolution of interest process, in this section we exploit the o.d.e. limit to study the effect of the threshold distribution.

\subsubsection{Uniform Distribution}
For the HILT model of interest evolution under the uniform threshold distribution, we can state the following theorem.

\begin{theorem}
Given the initial fraction of destinations $d_0$, in an HILT network with influence weight $\Gamma$ under uniform distribution, define $r = 1 - \Gamma + \Gamma d_0$. Then,
\begin{enumerate}
 \item The fluid limit for the evolution of interest is given by,
\begin{eqnarray}
\label{eqn:unif_explicit_1}
b(t) &=& \frac{d_0}{r} - \frac{d_0}{r} e^{-rt} \\
\label{eqn:unif_explicit_2}
d(t) &=& d_0 e^{-rt} 
\end{eqnarray}
\item The final fraction of destinations is given by 
\begin{eqnarray}
\label{eqn:d_0_b_infty}
b_\infty=\frac{d_0}{1-\Gamma+\Gamma d_0}
\end{eqnarray}
\end{enumerate}
\end{theorem}
\begin{proof}
The first part of the theorem is obtained by explicitly solving the o.d.e.s for the HILT model under uniform distribution (Equations~\ref{eqn:ODE_unif_1} and \ref{eqn:ODE_unif_2}) for the initial conditions $b(0)=0,d(0)=d_0$(see Appendix~\ref{app:ode-solving}).
\end{proof} 

Linear Threshold model under uniform distribution has been studied in discrete setting for general networks in \cite{kempe-etal03max-spread-infl}. Since in the HILT model all nodes are homogeneous, we will be interested in the influence of a set of size $k$. Consider the HILT network with $N$ nodes and influence weight $\gamma_N$. Let $I_{\gamma_N}(k)$ be the expected size of the \emph{terminal set of destinations} $\mathcal{A}(T)$, starting with $\mathcal{A}_0$ of size $k$ as the initial set of destinations. By using results from \cite{srini-kumar11LT-model-ncc}, we can show that, 
\[I_{\gamma_N}(k) = k[1 + (N-k)\gamma_N [ 1 + (N-k-1)\gamma_N[ 1 + \cdots \]

In the expression for $I_{\gamma_N}(k)$, noting $\lim_{N \rightarrow \infty} N \gamma_N = \Gamma$ and $d_0 = \frac{k}{N}$, we can show that as $N \rightarrow \infty$, $\frac{I_{\gamma_N}(k)}{N} \rightarrow b_\infty = \frac{d_0}{r}$ (refer Appendix~\ref{app:hilt-convergence}). This reconfirms the fact that the fluid model is consistent with the discrete formulation. Also, while \cite{srini-kumar11LT-model-ncc} allows us to compute only the final fraction of destinations, our current work provide a good approximation of the actual trajectories of the influence process.

\begin{figure}[t]
\centerline{\includegraphics[scale=0.3]{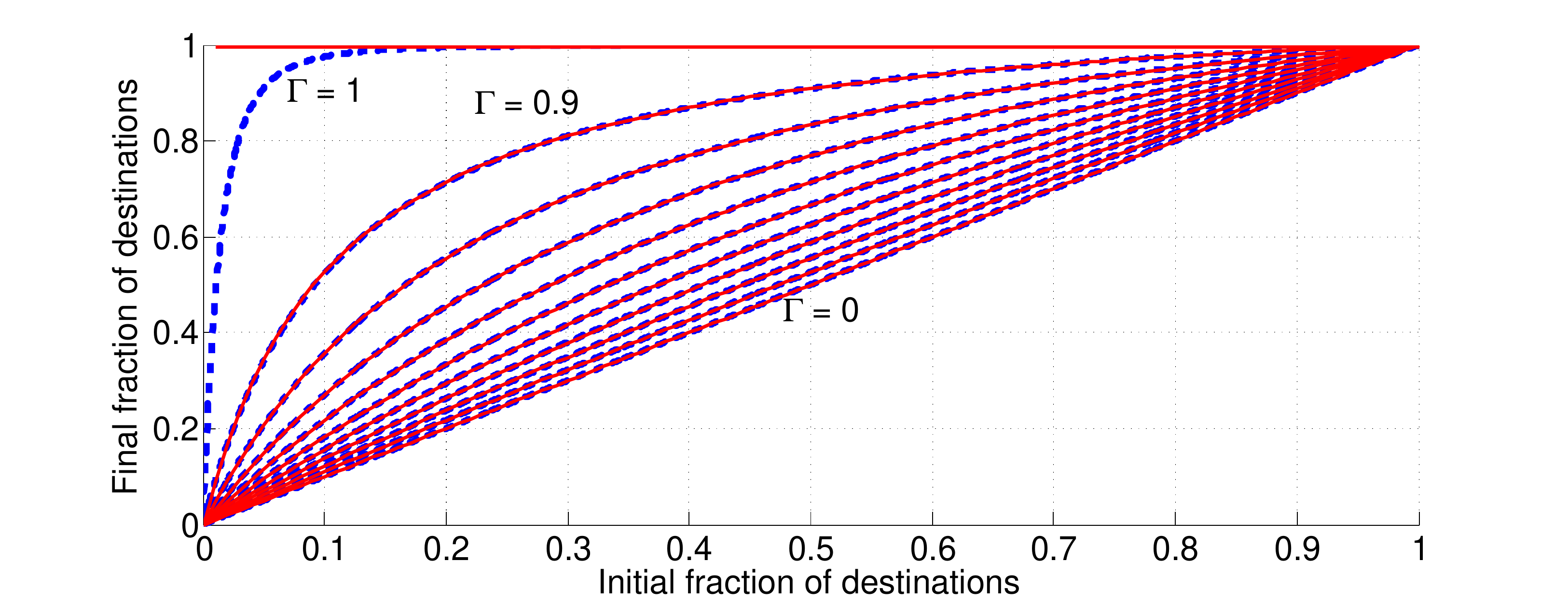}}
\vspace{-0.3cm}
\caption{Final fraction of destinations as predicted by the fluid limit $b_\infty$ (bold lines) and by the discrete formulation $\frac{I_{\gamma_N}(k)}{N}$ (dashed lines) plotted for various values of $\Gamma$ and $d_0$.}
\label{plot:hilt_k_3000}
\vspace{-0.5cm}
\end{figure}

\paragraph*{Remarks} Figure~\ref{plot:hilt_k_3000} shows the behaviour of $b_\infty$ for various of $\Gamma$ and $d_0$. Plotted alongside are the values of $\frac{I_{\gamma_N}(k)}{N}$ for the corresponding values of $\gamma_N$ and $k$ for $N=3000$ and it can be seen that the two solutions match really well. It is easy to note that when $\Gamma=0$, the HILT process does not have any effect on the number of destinations, i.e., the fraction of destinations and relays in the population remains a constant, similar to \cite{singh-etal11dtn-multi-destination}. We observe from the fluid limit that, as long as $\Gamma < 1$ we cannot influence the entire population (i.e., $b_\infty < 1$) unless we start off with the entire population to be destinations (i.e., $d_0 =1$). But if $\Gamma = 1$, then the influence of the destinations are at the maximum, and we can ultimately convert all nodes into destination, given we start with a non-zero fraction of destinations, i.e., $b_\infty =1$ provided $d_0 > 0$. 

\subsection{Exponential Distribution}
For the uniform distribution the hazard rate function is $h_F(x) = \frac{1}{1-x}$ which is monotonically increasing in $x$. Such a hazard rate implies that the population consists of relay nodes are more susceptible to getting interested in the content as the total number of destinations increases (i.e., total accumulated influence over the past increases). The exponential distribution with parameter $\Lambda$ yields $h_F(x)= \Lambda$, a constant hazard rate function. This implies a \emph{memoryless} property for the influence process, i.e., the relay nodes are equally likely to get influenced at a given time instant, irrespective of the net accumulated influence in the past. We can then state the following theorem.

\begin{theorem}
 Given the initial fraction of destinations $d_0$, in an HILT network with influence weight $\Gamma$ under exponential distribution with parameter $\Lambda$,
\begin{enumerate}
 \item The fluid limit of the evolution of interest is given by the solution to the o.d.e.,
\begin{eqnarray*}
 \dot{b} &=& d\\
 \dot{d} &=& \Lambda \Gamma d(1-b-d) - d 
\end{eqnarray*}
 \item The final fraction of destinations is the solution to the transcendental equation, 
\[ b_\infty =  1 - (1-d_0) e^{-\Lambda \Gamma b_\infty}\]
\end{enumerate}
\end{theorem}
\begin{proof}
The first part is obtained by substituting $h_F(x)= \Lambda$ in the HILT o.d.e. (Equations~\ref{eqn:ODE_F_1} and \ref{eqn:ODE_F_2}). This is equivalent to the SIR epidemic model with infection rate $\lambda \Gamma$ and recovery rate of 1 \cite{daley-gani99epidemic-modeling}. Note that $b(t)$ is then equivalent to the Recovered set (R), and $d(t)$ is equivalent to the Infected set (I). The second part follows from classic SIR literature \cite{kermack-mckendrick27SIR-epidemics} and observing that the basic reproduction number (expected number of new infections from a single infection) is $\Lambda \Gamma$. 
\end{proof}

\begin{figure}
\centerline{\includegraphics[scale=0.3]{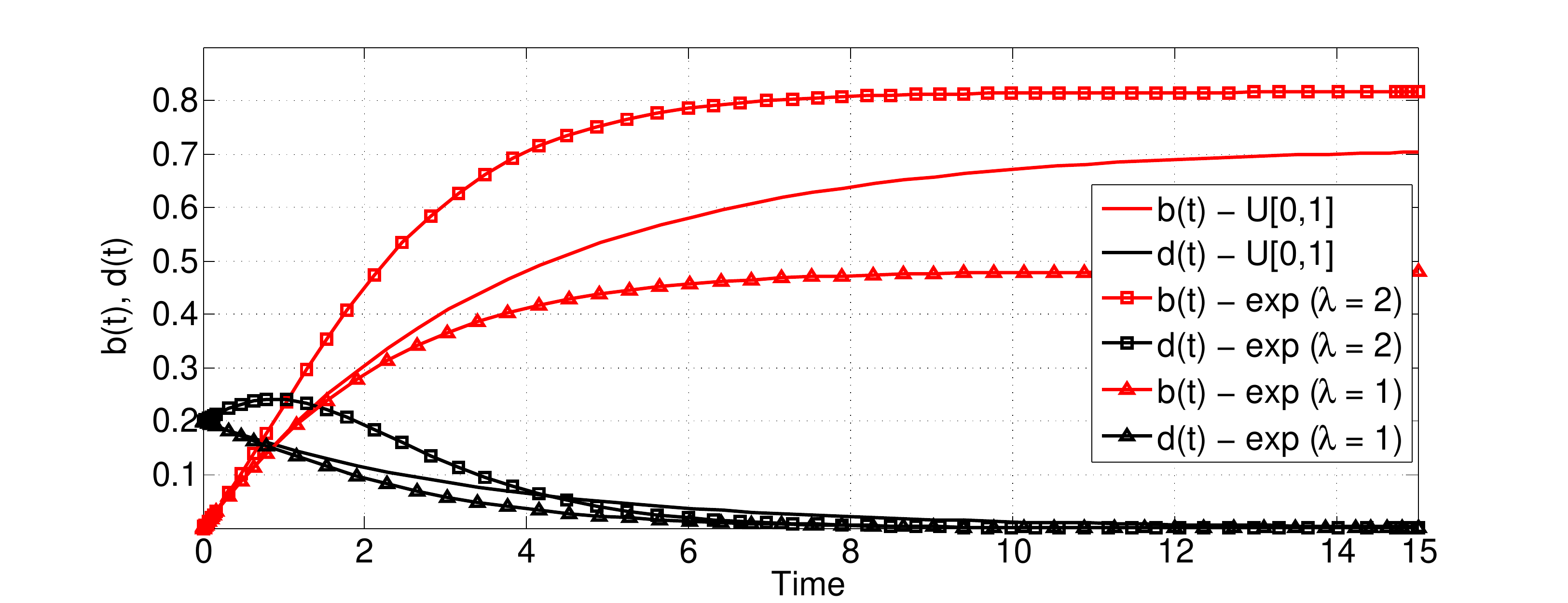}}
\vspace{-0.3cm}
\caption{Effect of threshold distribution in the HILT Model. Plotted are the interest evolution trajectories corresponding to Uniform distribution and Exponential distribution ($\lambda =1$ and $\lambda=2$).}
\label{plot:uniform-exponential}
\vspace{-0.5cm}
\end{figure}

\paragraph*{Remarks}
Figure~\ref{plot:uniform-exponential} shows the comparison of the interest evolution process $(b(t), d(t))$ with $\Gamma = 0.9$ and $d_0 = 0.2$, for uniform $[0,1]$ and exponential distribution with parameter $\Lambda = 1, 2$. Observe that, when the expected values of the thresholds are the same (uniform $[0,1]$ and exponential with $\Lambda =2$), the exponential distribution results in a higher value of $b_\infty$, since there will be a larger fraction of relay nodes with low sampled threshold. Also since $\Lambda \Gamma > \frac{1}{1-d_0}$, $\dot{d}(0) > 0$ and in SIR terminology\cite{kermack-mckendrick27SIR-epidemics} a \emph{proper outbreak} is said to have occurred. As $\Lambda$ decreases, we see that the expected influence thresholds of the relay nodes increase, and hence we notice a significant drop in the interest evolution process.

\section{Joint Evolution of Interest and Spread of the Content}
\label{sec:HILT-SI-fluid}
In the previous sections we studied the evolution of interest in the content. In this section, we shall adopt similar techniques to obtain the joint evolution of interest and spread of the content (HILT-SI model) in mobile P2P setting. Recall that interest evolution is modeled by the HILT model and the actual content transfer is modeled by a probabilistic copying process similar to the SI epidemic model. While the HILT model evolves independently, the evolution of the SI model depends upon the HILT model, since the copying process takes into account the relay/destination state of the receiving node, indicated by the \emph{want} bit of the node. In this section we will derive an o.d.e. limit for the SI part of the joint process. We shall also assume for simplicity that the thresholds in the HILT model are uniformly distributed; the analysis can be easily extended to HILT with general distribution $F$.

Pairwise meetings of the nodes consitute points of a Poisson process with rate $\lambda$, and let $\alpha_N = \frac{\alpha}{N}$ and $\sigma_N = \frac{\sigma}{N}$ be the copy probabilities to the destinations (nodes that are interested in the content) and relays (nodes not yet interested in the content) respectively. Recall from Section~\ref{sec:hilt-si-model} (Figure~\ref{fig:hilt-si}) $\mathcal{X}(k)$ the set of destinations that have the content and $\mathcal{Y}(k)$ the set of relays that have the content with $X(k)$ and $Y(k)$ the sizes of these sets respectively.  Taking into account the HILT driven conversion of relays that have the content into destinations, we note that both $X_k$ and $Y_k$ can be written in terms of the   binomial random variables $P_{(.)}(k)$ and $Q_{xy}(k)$, the first due to the content copying process (SI model) and the second due to the interest evolution process (HILT model). We can write, 
\[X(k+1) = X(k) + P_x(k) + Q_{xy}(k)\]
\[Y(k+1) = Y(k) + P_y(k) - Q_{xy}(k)\]
Since the nodes meet at rate $\lambda$ over a given slot, the probability that a destination node without the content, receives the content is $1 - e^{-\lambda \alpha_N (X(k) + Y(k))}$. Similarly for a relay node, the probability would be $1 - e^{-\lambda \sigma_N (X(k) + Y(k))}$. For large $N$, we can then write,
\[ P_x(k) \stackrel{\mathtt{dist.}}{=} \mathtt{Bin} \bigg( A(k) - X(k), \lambda \alpha_N (X(k) + Y(k)) \bigg)\]
\[ P_y(k) \stackrel{\mathtt{dist.}}{=} \mathtt{Bin} \bigg(N- A(k) - Y(k),\lambda \sigma_N (X(k) + Y(k)) \bigg) \]
Note also that due to the Interest Evolution process (HILT model) some relays get converted into destinations. Hence nodes from $\mathcal{Y}_k$ (relays that have the content) may transition to $\mathcal{X}_k$ (destinations that have the content). Since we are working with the uniform distribution, the probability of such a transition for each of the relay nodes would be $\frac{\gamma_N D(k)}{1 - \gamma_N B(k)}$. We denote by $Q_{xy}(k)$ the number of relay nodes that get converted during the $k$th time slot. Then,
\[ Q_{xy}(k) \stackrel{\mathtt{dist.}}{=} \mathtt{Bin} \bigg( Y(k), \frac{\gamma_N D(k)}{1 - \gamma_N B(k)} \bigg)\]
Hence we can write the evolution of the processes $X(k)$ and $Y(k)$ as follows.
\begin{eqnarray*}
X(k+1) - X(k) &=& \lambda \alpha_N (X(k) + Y(k)) (A(k) - X(k)) \\
&+& \frac{\gamma_N D(k)}{1 - \gamma_N B(k)} Y(k) + Z_x(k)\\ 
\end{eqnarray*}
\begin{eqnarray*}
Y(k+1) - Y(k) &=& \lambda \sigma_N (X(k) + Y(k))(N- A(k) - Y(k)) \\
&-& \frac{\gamma_N D(k)}{1 - \gamma_N B(k)} Y(k)+ Z_y(k)
\end{eqnarray*}
where $Z_x(k)$ and $Z_y(k)$ are zero mean noise variables respectively for the $X_k$ and $Y_k$ processes respectively. 
In order to obtain the fluid limit, we work with the scaled process which evolves over mini-slots of width $\frac{1}{N}$. In each mini slot, the probability that a destination node without the content, receives the content will then be $1 - e^{-\frac{\lambda \alpha_N}{N} (X(k) + Y(k))}$ and a similar scaling occurs for the relay case. The interest evolution scaling in the $Q_{xy}(k)$ term is similar to the probabilistic scaling adopted in Section~\ref{sec:ODE_interest}. We can then write the scaled processes $X^N(t)$ and $Y^N(t)$ and the corresponding fractional processes $\Tilde{X}^N(t)$ and $\Tilde{Y}^N(t)$ (see Appendix~\ref{app:hilt-si-fluidlimit} for details) and state the following theorem, along similar lines to Theorem~\ref{thm:kurtz-theorem-hilt}.

\begin{theorem}
\label{thm:kurtz-theorem-hilt-si}
Given the joint evolution Markov process $( \Tilde{B}^{N}(t), \Tilde{D}^{N}(t), \Tilde{X}^{N}(t), \Tilde{Y}^{N}(t) )$, we have for each $T > 0$ and each $\epsilon > 0$,
\begin{eqnarray*}
\mathbb{P} \bigg( \sup_{0 \leq u \leq T} \big| \big| \big( \Tilde{B}^{N}( \lfloor Nu \rfloor ),\Tilde{D}^{N}( \lfloor Nu \rfloor ) , \Tilde{X}^{N}( \lfloor Nu \rfloor ),\Tilde{Y}^{N}( \lfloor Nu \rfloor ) \big) \\
- \big( b(u),d(u),x(u),y(u) \big) \big| \big| > \epsilon \bigg) \stackrel{N\rightarrow \infty}{\rightarrow} 0
\end{eqnarray*}
where $(b(u),d(u),x(u),y(u))$ is the unique solution to the o.d.e.,
\begin{eqnarray}
\label{eqn:HILT-SI_1}
 \dot{b} &=& d\\
\label{eqn:HILT-SI_2}
 \dot{d} &=& \frac{\Gamma d}{1 - \Gamma b} (1-b-d) - d\\
\label{eqn:HILT-SI_3}
 \dot{x} &=& \lambda \alpha (x+y) (a-x)  + \frac{\Gamma d}{1 - \Gamma b} y \\
\label{eqn:HILT-SI_4}
 \dot{y} &=& \lambda \sigma (x+y) (1-a-y)  - \frac{\Gamma d}{1 - \Gamma b} y 
\end{eqnarray}
with initial conditions $(b(0)=0,d(0)=d_0, x(0)=x_0, y(0) = y_0)$.
\end{theorem}
\begin{proof}
A detailed proof verifying the conditions for Kurtz's theorem to hold, are presented in Appendix~\ref{app:kurtz-applied-hilt-si}.
\end{proof}
\subsection{Accuracy of the Fluid limit}
Figure~\ref{plot:want_have_scaled_ODE} shows the convergence of the scaled HILT-SI $(\Tilde{B}^N(\lfloor Nt \rfloor), \Tilde{D}^N(\lfloor Nt \rfloor), \Tilde{X}^N(\lfloor Nt \rfloor), \Tilde{Y}^N(\lfloor Nt \rfloor))$ process to the solutions of the above o.d.e $(b(t), d(t), x(t), y(t))$ for increasing values of $N=100,500,1000$. We see that for $N=1000$ the o.d.e approximates the scaled HILT-SI Markov chain very well.

The original HILT-SI process $(B(k), D(k), X(k), Y(k))$ is then compared with the solutions of the o.d.e. by superimposing multiple sample paths of the original discrete time HILT-SI process on the o.d.e solutions.(see Figure~\ref{plot:want_have_original_ODE}). We find that the o.d.e solution approximates the original process really well, and permits using the fluid limit approximation for sufficiently large $N$. 

\begin{figure}
\centerline{\includegraphics[width=9cm,height=4cm]{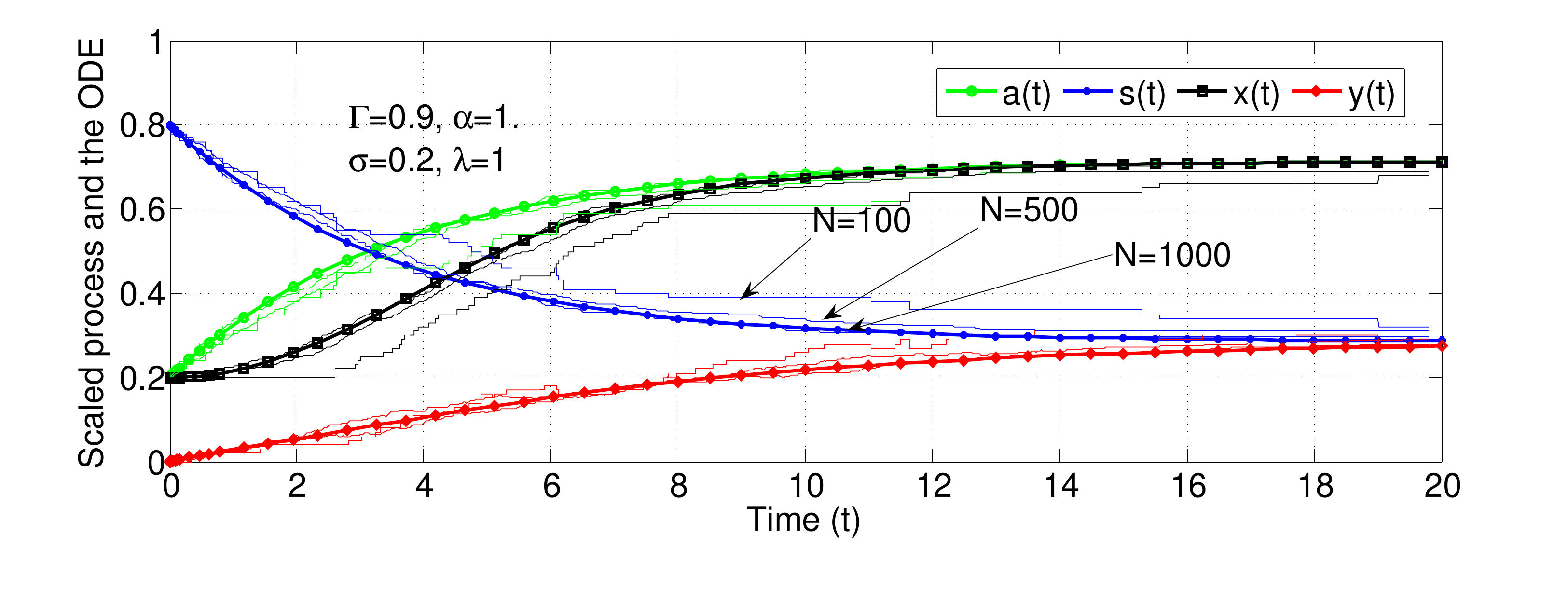}}
\vspace{-0.3cm}
\caption{Comparison of the scaled HILT-SI process for $N=50,100,1000$ with the corresponding fluid limits.}
\label{plot:want_have_scaled_ODE}
\vspace{-0.5cm}
\end{figure}

\begin{figure}
\centerline{\includegraphics[width=9cm,height=4cm]{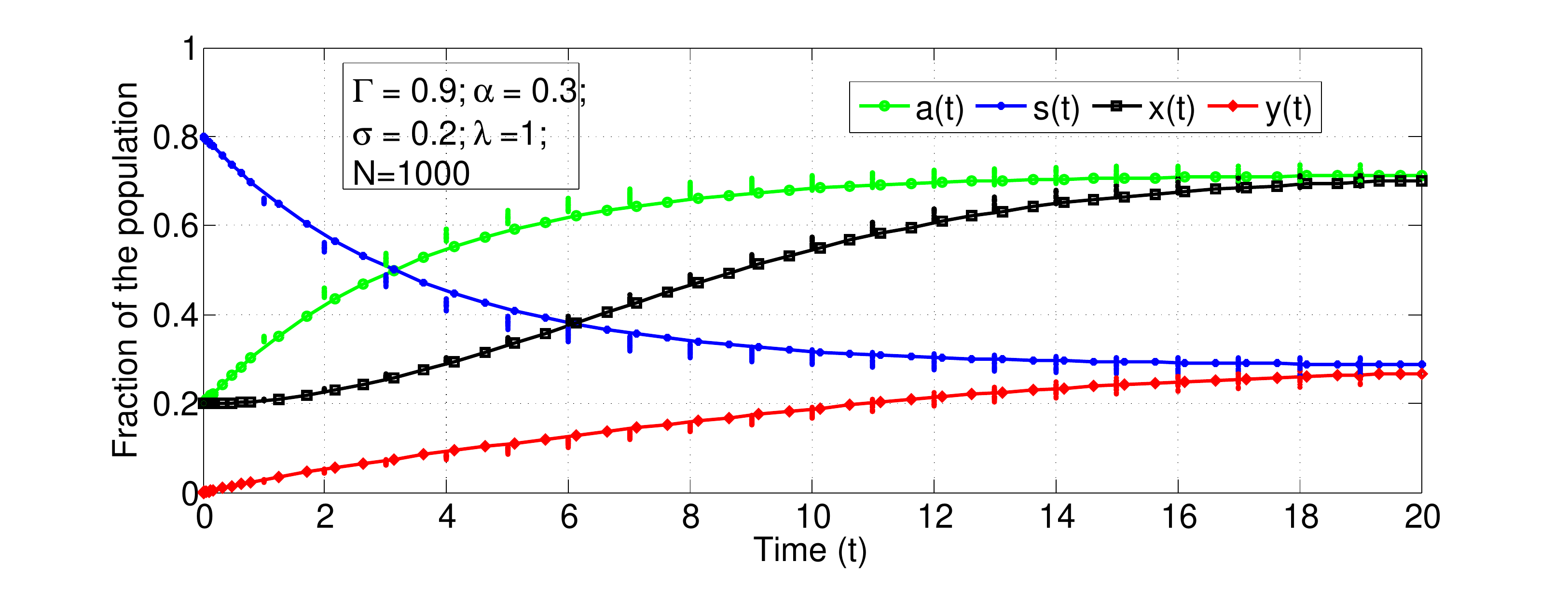}}
\vspace{-0.3cm}
\caption{Comparison of the unscaled HILT-SI process ($N=1000$) with the corresponding fluid limit.}
\label{plot:want_have_original_ODE}
\vspace{-0.5cm}
\end{figure}

\section{Numerical results}
\label{sec:numerical}
In this section we discuss the applications of the fluid limits through numerical results for various optimization problems of practical interest. While the first part of the section deals with optimizing for the interest evolution process considered in isolation, the latter part deals with optimizations for the joint evolution process. As an example, we work with HILT under uniform threshold distribution in this section.

\subsection{Interest Evolution}
Content creators are often interested in understanding the evolution of popularity of their content, and would wish to maximize the level of popularity achieved. This, in our model, is equivalent to the final fraction of destinations (nodes that are interested in receiving the content). Once the content is created, in most cases, the content creator does not have control over influence weight $\Gamma$ or the threshold distribution $F$ of the population. In order to increase the spread of interest in the content, the only parameter that can be controlled is $d_0$, the initial fraction of destinations in the population. In Section~\ref{sec:HILT-fluid}, we derived the relation between $d_0$ and $b_\infty$, the final fraction of destinations. We might be interested in choosing the right $d_0$ which can give us the required $b_\infty$, and we see that by rearranging Equation (\ref{eqn:d_0_b_infty}) we get,
\begin{eqnarray}
\label{eqn:b_infty_d_0}
 d_0 = \frac{b_\infty (1 -\Gamma)}{1 - b_\infty \Gamma} 
\end{eqnarray}
Since the o.d.e. provides a good approximation for the temporal evolution of interest, we can also obtain results for time-constrained spread of popularity, i.e., the time taken by the process for the spread of influence is also a constraint, in addition to the initial fraction of destinations.
\begin{theorem}
Given the initial fraction of destinations $d_0$ in an HILT network with parameter $\Gamma$, the time we have to wait to get the final fraction of destinations to be at least $\beta$ ($\beta < \frac{d_0}{r}$) is given by,
\[ T(\beta, d_0, \Gamma) = \frac{1}{r} \ln \bigg( \frac{1-r}{1-\frac{\beta}{d_0}r} \bigg) \]
where $r = 1 - \Gamma + \Gamma d_0$.
\end{theorem}

\begin{proof}
Firstly, note that since $a_\infty = b_\infty = \frac{d_0}{r}$, we cannot reach $\beta > \frac{d_0}{r}$. Since we are observing the process at a finite time $T$, $d(T)$ is not zero. Hence, we consider $a(T)=b(T)+d(T)$ and set it to $\beta$. We get,
\begin{eqnarray}
\label{eqn:hilt_d_0_finite_T}
a(T) = d_0 \big( \frac{1}{r} - ( \frac{1}{r}-1) e^{-rT} \big) = \beta
\end{eqnarray}
Rearranging terms,we get the expression for $T(\beta, d_0, \Gamma)$. 
\end{proof}
\begin{figure}
\centerline{\includegraphics[scale=0.3]{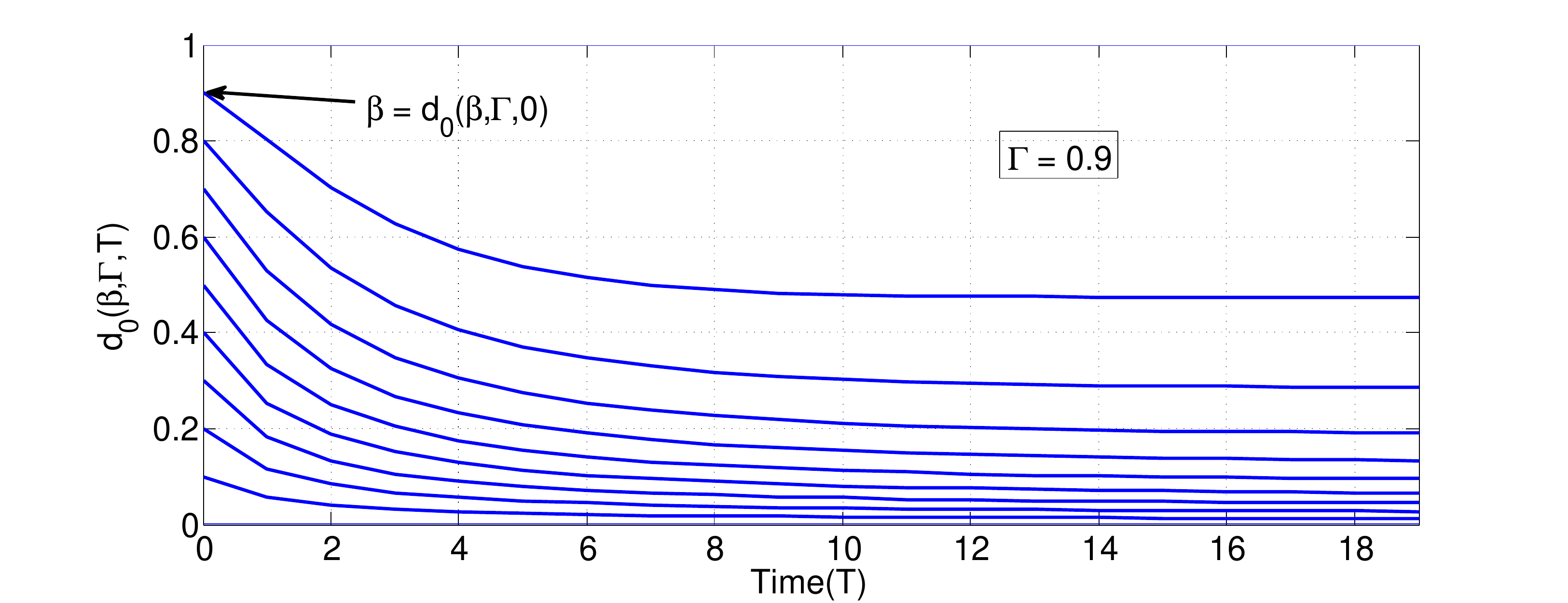}}
\vspace{-0.3cm}
\caption{Variation of $d_0^{\star}$ versus $T$ for various values of $\beta$ with $\Gamma=0.9$}
\label{plot:d_0_fixed_Gamma_varying_beta_versus_T}
\vspace{-0.5cm}
\end{figure}

\begin{figure}
\centerline{\includegraphics[scale=0.3]{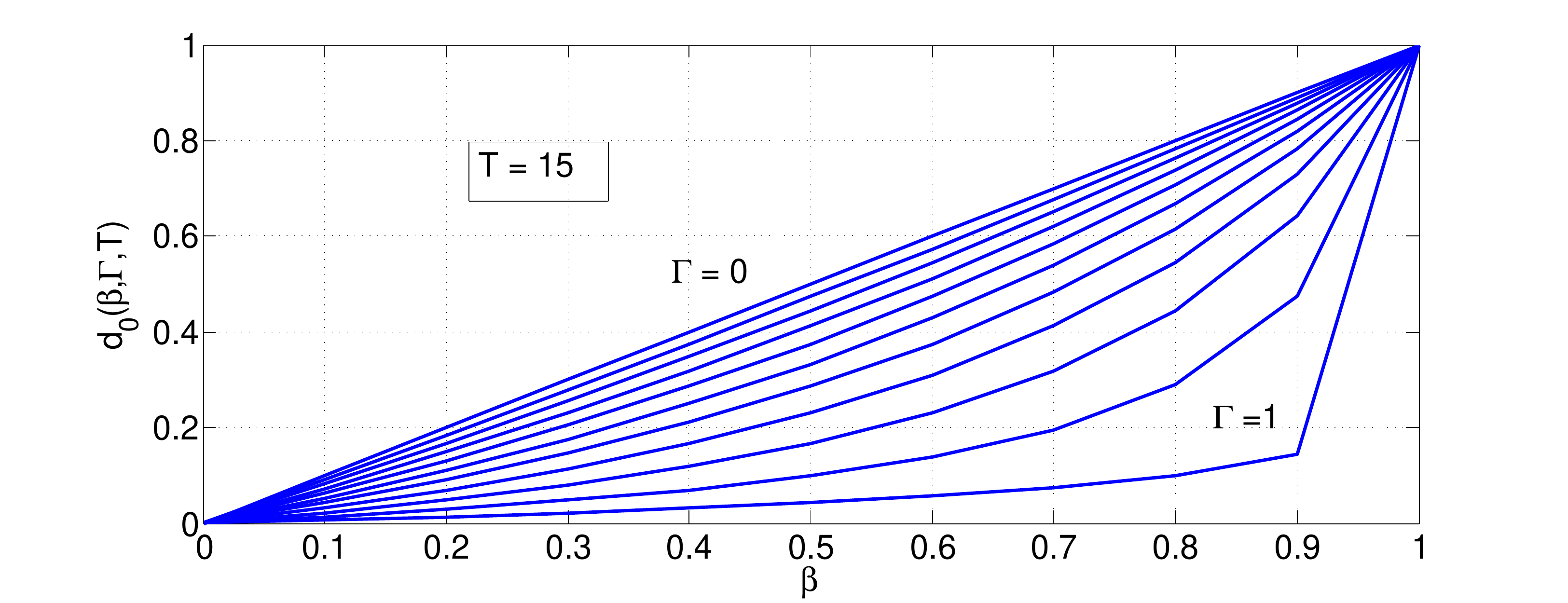}}
\vspace{-0.3cm}
\caption{Variation of $d_0^{\star}$ versus $\beta$ for various values of $\Gamma$ at $T=15$}
\label{plot:d_0_fixed_T_varying_Gamma_versus_beta}
\vspace{-0.5cm}
\end{figure}

Another interesting question would be to determine the initial fraction $d_0$ to be chosen so that by time $T$ we will have at least $\beta$ fraction of the nodes in the destination set, in the HILT network with parameter $\Gamma$. This can be solved numerically using the following fixed point equation obtained from Equation (\ref{eqn:hilt_d_0_finite_T}).
\[ e^{-rT} = \frac{1-\frac{\beta}{d_0} r}{1-r} \]
Let $H(d_0) = e^{-rT}$ and $G(d_0)=\frac{1-\frac{\beta}{d_0} r}{1-r}$. We know that $d_0^{\star}$ that solves $H(d_0)=G(d_0)$ will lie in $[\frac{\beta (1 -\Gamma)}{1 - \beta \Gamma}, 1]$ and that the solution is unique, since $a(T)$ is a monotonic function in $d_0$. We also know that for $d_0 < d_0^{\star}$, $H(d_0) > G(d_0)$ and for $d_0 > d_0^{\star}$, $H(d_0) < G(d_0)$. Under the above conditions, we can use the iterative bisection method that will converge to $d_0^{\star}$. 

The variation of $d_0^{\star}$ with respect to the parameters $\beta$, $\Gamma$ and $T$ can be seen in Figures~\ref{plot:d_0_fixed_Gamma_varying_beta_versus_T} and \ref{plot:d_0_fixed_T_varying_Gamma_versus_beta}. It is clear that,
\begin{itemize}
\item For fixed $\Gamma$ and $T$, a higher value of $\beta$ requires a higher $d_0$ (Figure~\ref{plot:d_0_fixed_T_varying_Gamma_versus_beta}). 
\item If the network has a high $\Gamma$, then it is sufficient to start with a smaller $d_0$ to reach a given $\beta$ by $T$. (Figure~\ref{plot:d_0_fixed_T_varying_Gamma_versus_beta}).
\item As the fixed time of interest $T$ is increased (the time constraint is relaxed), the required initial $d_0$ decreases, from $\beta$ at $T=0$ and finally reaches $d_0$ obtained by setting $b_\infty = \beta$ in Equation (\ref{eqn:b_infty_d_0}) (Figure~\ref{plot:d_0_fixed_Gamma_varying_beta_versus_T}).  
\end{itemize}

\subsection{Joint evolution of interest and spread}
In this section we discuss the optimization problems that might be of interest for the joint evolution process. Given that the network parameters $\Gamma$ and $d_0$ are fixed, interest in the content evolves independently. The main motive of the content provider would then be to ensure that the content is delivered to as many destinations as possible.

Recall the o.d.e.s for the HILT-SI model given by Equations~\ref{eqn:HILT-SI_1}-\ref{eqn:HILT-SI_4}. Since we intend to deliver to the destinations, it would be optimal to set $\alpha =1$. Hence the only control parameter is the copy probability to a relay, $\sigma$. This may be controlled by incentivizing or penalizing copying to relay nodes. As noted earlier, copying to a relay has two advantages, since the relay might meet a destination in the future, or might itself get converted into one, by the HILT process. But having a high $\sigma$ value, will lead to increase in the number of \emph{wasted} copies, the number of relays (not interested) that will end up having the content. The joint evolutions o.d.e.s have been simulated numerically for $\alpha=1$ and two different values of $\sigma$ and are plotted in Figure~\ref{plot:want_have_diff_sigma}. As discussed, having a higher value of $\sigma$ accelerates the $x(t)$ process, but leads to higher value of $y(t)$. Hence there is a tradeoff between achieving a high $x(t)$ while keeping $y(t)$ under bounds. We discuss two possible optimization problems for $\sigma$ that can be posed, keeping in mind this tradeoff between $x(t)$ and $y(t)$. To understand clearly the effect of $\sigma$ on the copying process, we work with $\Gamma =0.9$ and have fixed the initial fraction of destinations $d_0 = 0.2$, initial fraction of destinations that have the content $x_0 = 0.2$. This implies that the set of nodes initially interested in the content is equal to the set of nodes that have the content. We assume that the content is delivered a priori to these initial destinations by some other means. 
\begin{figure}
\centerline{\includegraphics[width=9cm,height=4cm]{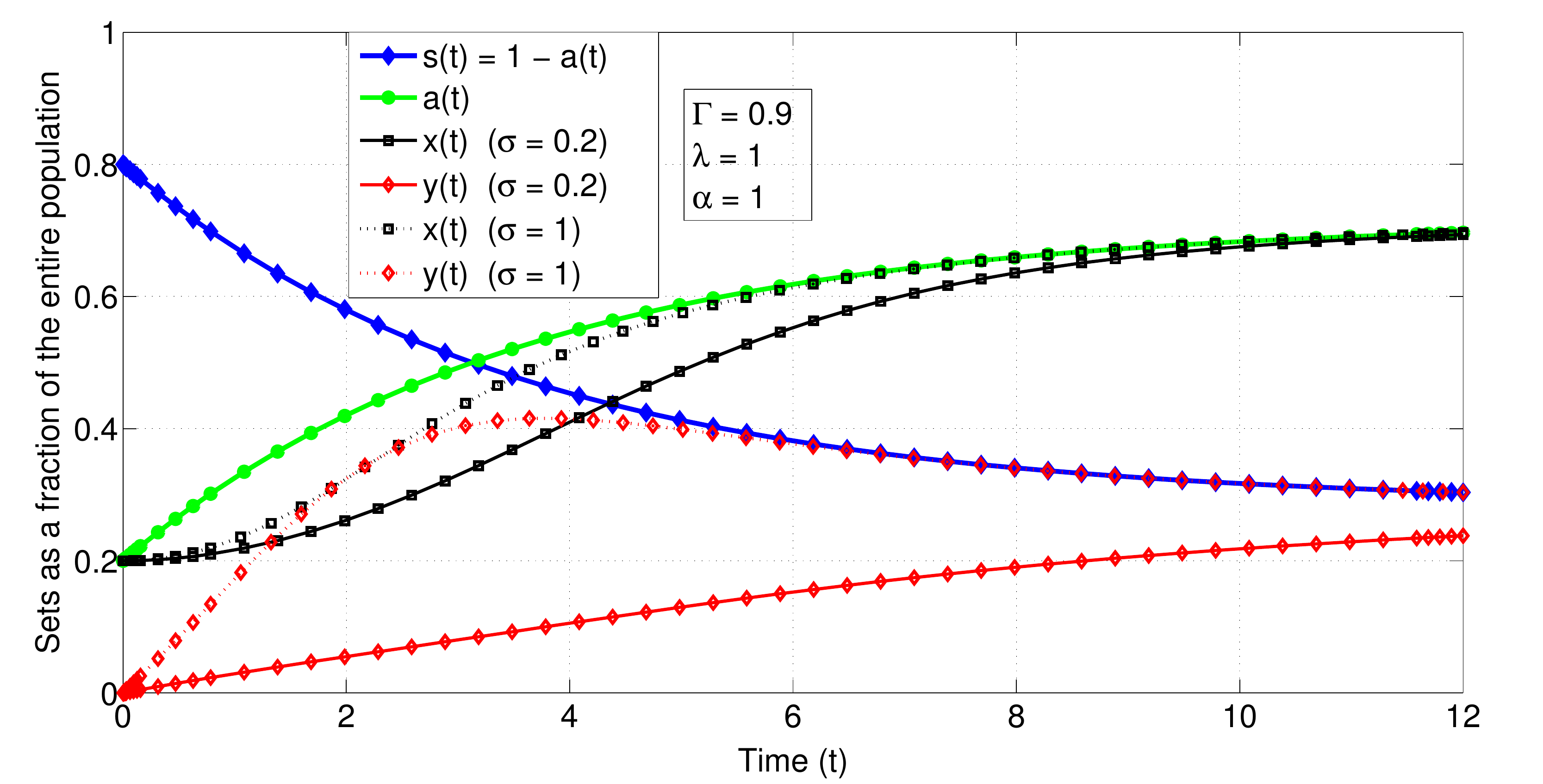}}
\vspace{-0.3cm}
\caption{Numerical simulation of HILT-SI model for two different values of relay copy probability $\sigma=1$ and $\sigma=0.2$.}
\label{plot:want_have_diff_sigma}
\vspace{-0.5cm}
\end{figure}

\subsubsection{Maximize target spread}
As content providers, we might be interested to deliver the content to as many destinations as possible by some fixed time. This might be the case,when the content is time-dependent and its usefulness expires by that fixed time. The fraction of destinations is itself increasing with time, due to the HILT interest evolution process. We would also wish to ensure that the number of relays that have the content, by that fixed time, is minimal. Our aim is then to maximize the value of $x(t)$ at some given fixed $\tau$, subject to keeping $y(\tau)$ under control. The problem can be formally posed as,
\[ \max_{\{ \sigma: y(\tau) \leq \zeta\}} x(\tau) \]

Numerically, we can compute the feasible set $\{ \sigma: y(\tau) \leq \zeta\}$ obtain the optimal solution $\sigma^\star$ that maximizes $x(\tau)$. We also observe experimentally that for a fixed $\tau$, $x(\tau)$ and $y(\tau)$ are monotonically increasing with $\sigma$. 

Figure~\ref{plot:opt_target_size_varying_tau} shows the variation of $\sigma^\star$ and the corresponding optimal $x_(\tau)^{\star}$ for various values of $\tau$ and a fixed fraction $\zeta$. Observe that, for $\tau$ small, we can afford a high value of $\sigma$ since $y(t)$ cannot exceed $\zeta$ in such a short time. Also, when $\tau$ is large, we can have a high $\sigma$, provided the interest evolution process is strong enough (high $\Gamma$). This will ensure that the number of relays $s(t)$ will be low enough by $\tau$, and since $y(t) \leq s(t)$, $\forall t$, $y(t)$ cannot exceed $\zeta$. 

\begin{figure}
\centerline{\includegraphics[scale=0.3]{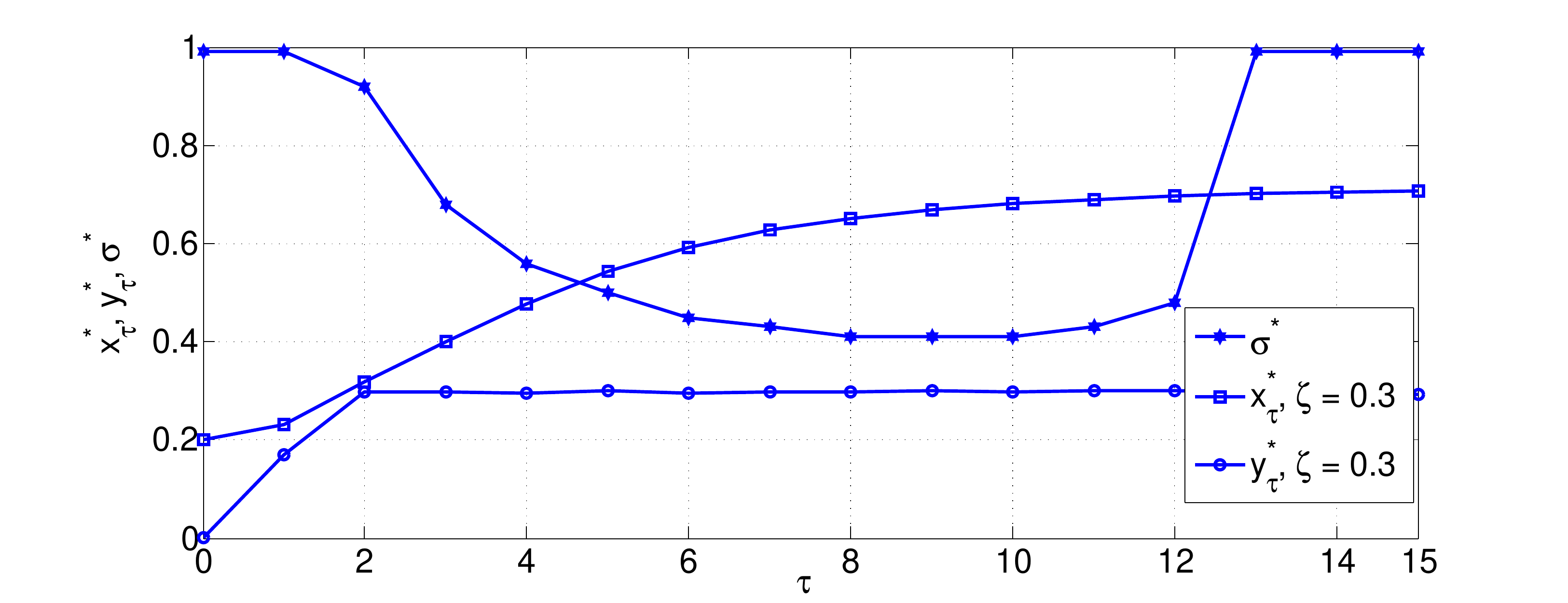}}
\vspace{-0.3cm}
\caption{Maximize target spread: The optimal solution plotted for a fixed value of $\zeta$ and varying $\tau$.}
\label{plot:opt_target_size_varying_tau}
\vspace{-0.5cm}
\end{figure}
\begin{figure}
\centerline{\includegraphics[scale=0.3]{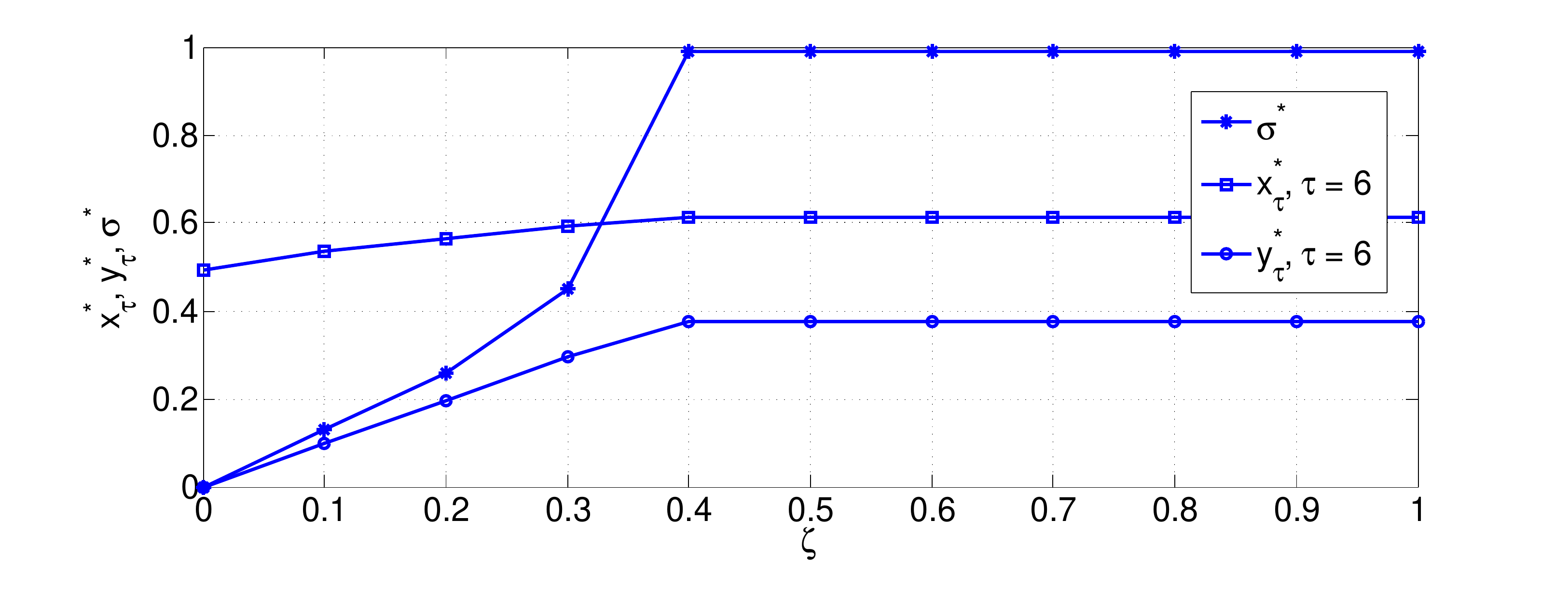}}
\vspace{-0.3cm}
\caption{Maximize target spread: The optimal solution plotted for a fixed value of $\tau$ and varying $\zeta$.}
\label{plot:opt_target_size_varying_zeta}
\vspace{-0.5cm}
\end{figure}

Figure~\ref{plot:opt_target_size_varying_zeta} shows the similar variations for various fractions $\zeta$ and a fixed time $\tau$. As $\zeta$, the constraint on the number of relay nodes with the content at $\tau$ is relaxed, we see that the optimal $\sigma^{\star}$ increases. It can be seen that $y(\tau)$ stabilizes at the value dictated by the number of relay nodes $s(\tau)$ for high value of $\sigma$. Also note that, if the fixed time of interest is high, then there is negligible contribution due to increase in $\sigma$ since $x(t)$ increases to $a(t)$ eventually irrespective of $\sigma$ (see Figure~\ref{plot:want_have_diff_sigma}).

\subsubsection{Minimize Reach time}
Another problem of interest is to deliver the content to a fraction of the destinations as early as possible. This might be the case when, delivering to a considerable fraction of destinations accrues some benefit for the content provider from the content creator, and would like to earn the benefit as early as possible.  Hence we try to minimize the time taken to reach a particular value of $x(t)$. Define $\tau_\eta = \inf \{t : x(t) \geq \eta \} $. Then we can formally pose this problem as,
\[ \min_{\{\sigma: y(\tau_\eta) \leq \zeta \}} \tau_\eta \]

We numerically compute the feasible set by computing $\tau_\eta$ for each $\sigma$ and verifying if $y(\tau_\eta) \leq \zeta$.  We experimentally observe that for fixed $\eta$, $\tau_\eta$ is monotonic decreasing in $\sigma$ and $y(\tau_\eta)$ is monotonic increasing in $\sigma$.

Figures~\ref{plot:opt_reach_time_varying_zeta} and \ref{plot:opt_reach_time_varying_eta} show the numerically evaluated optimal solutions for the above optimization problem, for various values of $\eta$ and $\zeta$. In Figure~\ref{plot:opt_reach_time_varying_zeta}, it is clear that a higher value of $\zeta$ allows us to use a high value of $\sigma$ so that we can accelerate the content delivery to destinations. As a result, we also see that the time taken to deliver the content decreases with increasing $\zeta$.

\begin{figure}
\centerline{\includegraphics[scale=0.3]{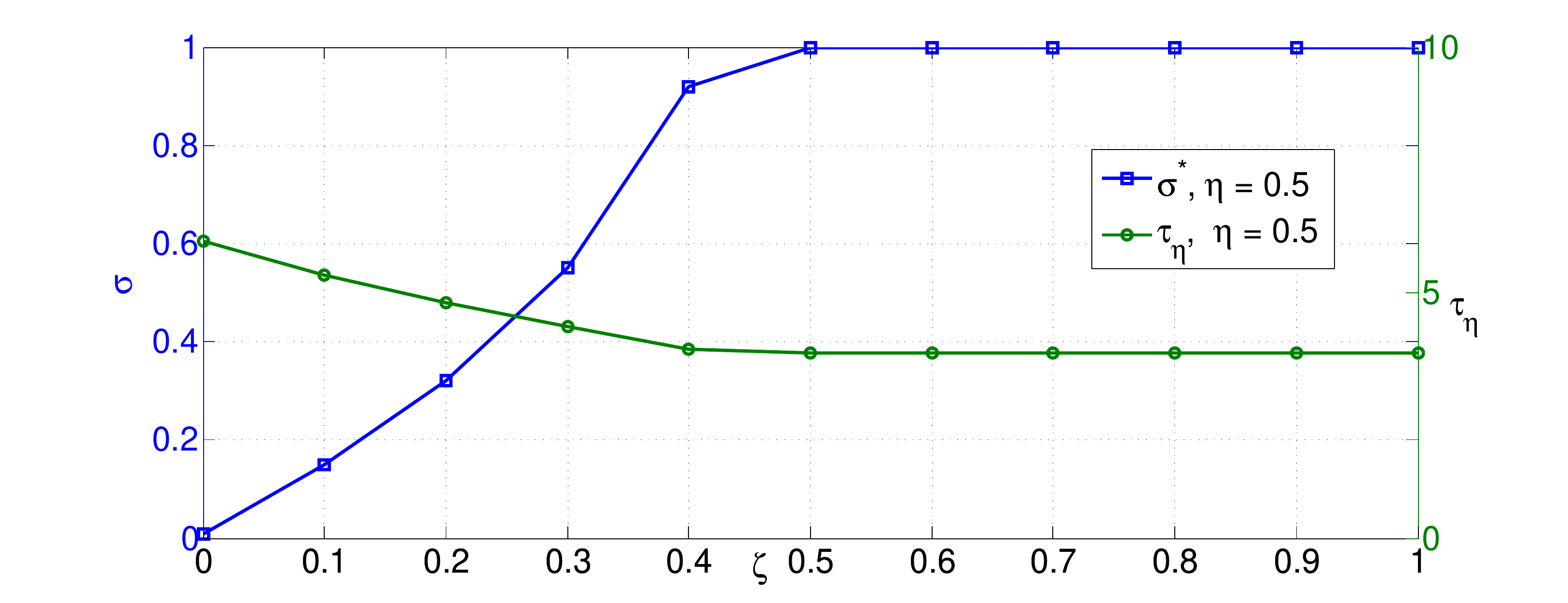}}
\vspace{-0.3cm}
\caption{Minimize reach time: The optimal solution plotted for a fixed value of $\eta$ and varying $\zeta$.}
\vspace{-0.5cm}
\label{plot:opt_reach_time_varying_zeta}
\end{figure}

\begin{figure}
\centerline{\includegraphics[scale=0.3]{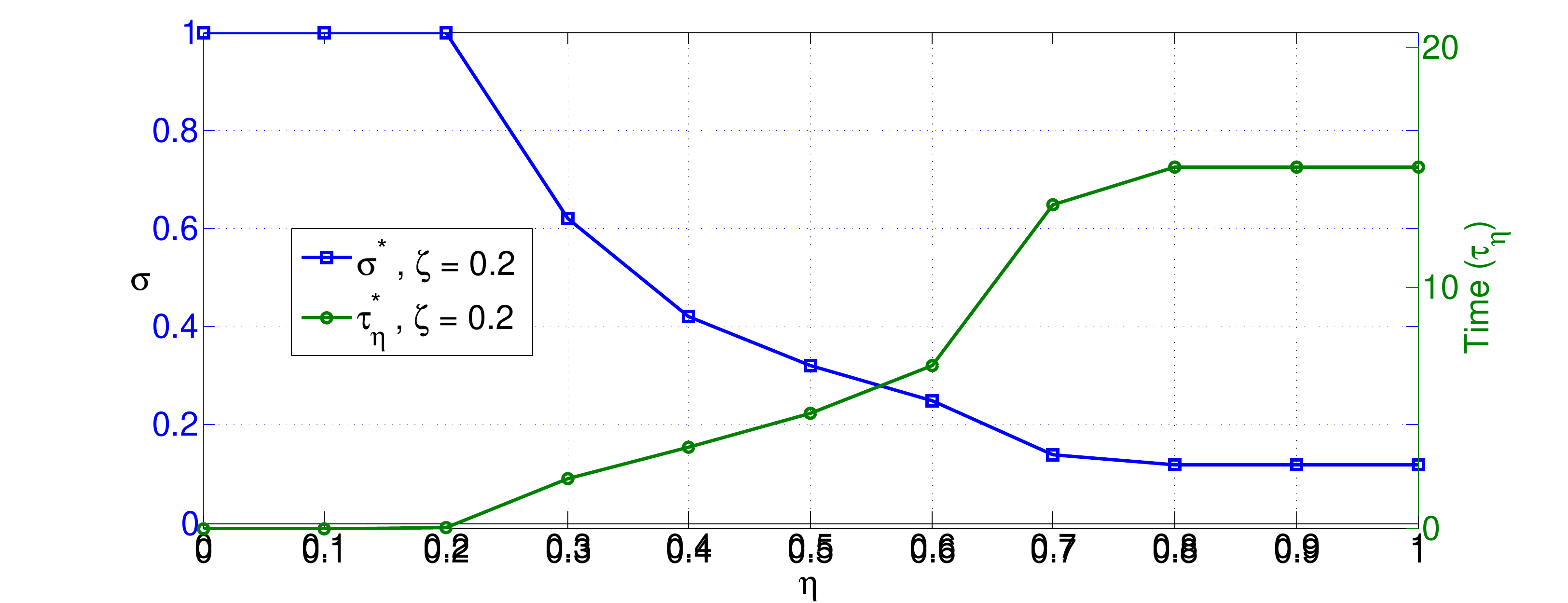}}
\vspace{-0.3cm}
\caption{Minimize reach time: The optimal solution plotted for a fixed value of $\zeta$ and varying $\eta$.}
\vspace{-0.5cm}
\label{plot:opt_reach_time_varying_eta}
\end{figure}

In Figure~\ref{plot:opt_reach_time_varying_eta}, note that since we begin with $x_0 = 0.2$, for $\eta \leq 0.2$, $\tau_\eta = 0$ and $\sigma$ can be any value. When $\eta$ increases, as long as $\tau_\eta$ is low enough, we can afford to have a high $\sigma$ and still keep $y(\tau_\eta)$ under control, since the $y(t)$ process may not have enough time to reach $\zeta$. But as $\eta$ increases further, $\sigma$ needs to be reduced, since an increased value of $\tau_\eta$ implies, the process will run longer and $y(\tau_\eta)$ might then exceed $\zeta$.

\section{Conclusions}
In this paper we studied the joint evolution of popularity and spread of content in a mobile P2P environment. We used two different models to capture the processes of interest evolution (HILT model) and content delivery (SI model), and derive their fluid limits. We then showed that the classic SIR epidemic model can be derived as a special case of the HILT model and derived explicit solutions for the model when the thresholds were uniformly distributed. We then used the fluid limits to address several optimization problems that might be of practical interest. This work can be extended in several possible directions. The co-evolution problem can be studied for other related optimization problems, and we can try obtaining explicit optimal solutions for a few special cases. One could generalize the problem for analysing the spread of popularity and spread for multiple P2P content. And finally one could consider other models for both interest evolution and for the copying process, including those where the processes of interest evolution and content delivery are dependent on each other.
\bibliographystyle{IEEEtran}
\bibliography{bib-infocom12}

\appendices
\section{Interest Evolution is a Markov Chain}
\label{app:HILT-markov}

\begin{proposition} \label{prop:hilt_dtmc}
 For the HILT model, $(B(k), D(k)), k \geq 0,$ is a discrete time Markov chain (DTMC).
\end{proposition}
\proof Since, for each $k \geq 0$, $B(k+1) = B(k) + D(k)$, it suffices
to show that, for $k \geq 0$, 
\begin{eqnarray*}
  P(D(k+1) = \ell | (B(0), D(0)), (B(1), D(1)), \cdots, ((B(k), D(k))= (j,m)))
\end{eqnarray*}
is a function only of $(j,m)$. We need the following simple lemma (in
which the new notation is local to the lemma).

\begin{lemma} \label{lem:recursion_of_threshold_cdf}
  Given $(X_1, X_2, \cdots, X_n)$ i.i.d.\ random variables with common
  c.d.f.\ $G(x), x \geq a,$ where $a > 0$, and given $b > a$, define
  \begin{eqnarray*}
    \mathcal{Z} = \{i: 1 \leq i \leq n, X_i > b\}, \ \mathrm{and} \ Z = |\mathcal{Z}|.
  \end{eqnarray*}
  Let $Y_1, Y_2, \cdots, Y_Z$ be random variables jointly distributed
  as $X_i, i \in \mathcal{Z}$ (ordered in the increasing sequence of
  indices). Then
  \begin{eqnarray*}
    P(Y_i \leq y_i, 1 \leq i \leq Z | Z = m) = \Pi_{i=1}^m \frac{G(y_i) - G(b)}{1 - G(b)}
  \end{eqnarray*}
  for $y_i \geq b, 1 \leq i \leq m.$
\end{lemma}
\proof (of Lemma~\ref{lem:recursion_of_threshold_cdf}) It is easily
seen that
\begin{eqnarray*}
P(Z=m, Y_1 \leq y_i, \cdots, Y_m \leq y_m) = {n \choose m} \  (G(b))^{n-m} \  \Pi_{i=1}^m (G(y_i) - G(b))
\end{eqnarray*}
and
\begin{eqnarray*}
  P(Z=m) = {n \choose m} \  (G(b))^{n-m} \  (1 - G(b))^m
\end{eqnarray*}
from which the desired result immediately follows. $\blacksquare$

Returning to the proof of Proposition~\ref{prop:hilt_dtmc}, given
$(B(0), D(0)), (B(1), D(1)), \cdots, ((B(k), D(k))= (j,m))$, via a
recursive application of Lemma~\ref{lem:recursion_of_threshold_cdf}
(starting from the initial i.i.d.\ thresholds $\Theta_i, 1 \leq i \leq
N,$ with common c.d.f.\ $F(\theta)$), we conclude that the thresholds
of the users in $\mathcal{N} \ \backslash \ (\mathcal{B}(k) \cup
\mathcal{D}(k))$ are i.i.d.\ with common c.d.f.\ $\frac{F(\theta) -
  F(\gamma j)}{1 - F(\gamma j)}$, over the range $\theta \geq \gamma
j$. At the end of period $k$, the newly interested users in
$\mathcal{D}(k)$ will serve as additional influence on the users in
$\mathcal{N} \ \backslash \ (\mathcal{B}(k) \cup \mathcal{D}(k))$. Of
these, $\ell$ will become interested (i.e., $D(k+1) = \ell$) with
probability
\begin{eqnarray*}
  {(N - (j+m)) \choose \ell} 
  \left(\frac{F(\gamma(j+m)) - F(\gamma j)}{1 -  F(\gamma j)}\right)^\ell
  \left(1 - \frac{F(\gamma(j+m)) - F(\gamma j)}{1 -  F(\gamma j)}\right)^{(N - (j+m)) - \ell}
\end{eqnarray*}
which depends on the ``history'' only via $(j,m)$, and thereby establishes the
desired result. $\blacksquare$

\section{Fluid Limit for the HILT Model}
\label{app:hilt-fluidlimit}
Recalling from Section~\ref{sec:ODE_interest}, the evolution of $(\Tilde{B}^{N}(t), \Tilde{D}^{N}(t))$ can be written in terms of its mean ``drift'' at any $t$ as follows.
\[ \Tilde{B}^{N}(t+1) = \Tilde{B}^{N}(t) + \frac{\Tilde{D}^{N}(t)}{N} + \Tilde{Z}^{N}_b(t+1) \]
\begin{eqnarray*}
\lefteqn{\Tilde{D}^{N}(t+1) = \frac{N-1}{N} \Tilde{D}^{N}(t)}\\
\hspace{-2cm} &+& \mathbb{E} \bigg[ \frac{F(\gamma_N N (\Tilde{B}^{N}(t) + \Tilde{C}^{N}(t)))- F(\gamma_N N (\Tilde{B}^{N}(t))}{1-F(\gamma_N N (\Tilde{B}^{N}(t))} \bigg]\times\\
& & (1-\Tilde{B}^{N}(t) -\Tilde{D}^{N}(t)) + \Tilde{Z}^{N}_d(t+1)
\end{eqnarray*}
where we have defined $\Tilde{Z}^{N}_b(t)$ and $\Tilde{Z}^{N}_d(t)$ in an analogous manner. The mean drift rate function, $g^{N}(\Tilde{B}^{N}(k),\Tilde{D}^{N}(k))$ per unit update step size $\frac{1}{N}$ becomes,
\[\frac{g^{N}(\Tilde{B}^{N}(t),\Tilde{D}^{N}(t))}{\frac{1}{N}} = \bigg( g^{N}_1(.), g^{N}_2(.) \bigg) \]
where 
\[g^{N}_1(.)= \Tilde{D}^{N}(t) \]
\begin{eqnarray*}
\lefteqn{g^{N}_2(.)} \\
&=&  N \mathbb{E} \bigg[ \frac{F(\gamma_N N (\Tilde{B}^{N}(t) + \Tilde{C}^{N}(t)))- F(\gamma_N N (\Tilde{B}^{N}(t))}{1-F(\gamma_N N (\Tilde{B}^{N}(t))} \bigg] \times \\ 
& &\bigg( 1-\Tilde{B}^{N}(t)- \Tilde{D}^{N}(t) \bigg) - \Tilde{D}^{N}(t) 
\end{eqnarray*}
Recall $\gamma_N \times (N-1) = \Gamma$ and $f(x)$ the density function of the threshold distribution $F$. We can then show that (see Appendix~\ref{app:limit_expectation}), 
\begin{eqnarray*}
\lim_{N \rightarrow \infty} N \mathbb{E} \bigg[ \frac{F(\gamma_N N (\Tilde{B}^{N}(t) + \Tilde{C}^{N}(t)))- F(\gamma_N N (\Tilde{B}^{N}(t))}{1-F(\gamma_N N (\Tilde{B}^{N}(t))} \bigg]\\
= \frac{f(\Gamma b)\Gamma d}{1- F(\Gamma b)} \\
\end{eqnarray*}
Letting $g_1(b,d)= d$ and $g_2(b,d) =  \frac{f(\Gamma b)\Gamma d}{1- F(\Gamma b)} (1-b-d) - d$ and define,
\[ g(b,d) := \bigg( g_1(b,d),g_2(b,d) \bigg) \]
We can then use these limiting drift functions to obtain the fluid limits corresponding to the HILT model.

\section{}
\label{app:limit_expectation}
Consider 
\[ lim_{N \rightarrow \infty} N \mathbb{E} \bigg[ \frac{F(\gamma N (\Tilde{B}^{N}(t) + \Tilde{C}^{N}(t)))- F(\gamma N (\Tilde{B}^{N}(t))}{1-F(\gamma N (\Tilde{B}^{N}(t))} \bigg]\] 
where the expectation is with respect to $\Tilde{C}^N(t)$ given $(\Tilde{B}^{N}(t), \Tilde{D}^{N}(t))$. Applying Taylor's expansion to $F(\gamma N (\Tilde{B}^{N}(t) + \Tilde{C}^{N}(t)))$ around $\Tilde{B}^{N}(t)$ we get,

\begin{eqnarray*}
\lefteqn{ lim_{N \rightarrow \infty} N \mathbb{E} \bigg[ \frac{F(\gamma N (\Tilde{B}^{N}(t) + \Tilde{C}^{N}(t)))- F(\gamma N (\Tilde{B}^{N}(t))}{1-F(\gamma N (\Tilde{B}^{N}(t))} \bigg]}\\
&=& lim_{N \rightarrow \infty} N \mathbb{E} \bigg[ \frac{ \gamma N \Tilde{C}^N(t) f(\gamma N \Tilde{B}^N(t))}{1-F(\gamma N (\Tilde{B}^{N}(t))} \bigg] \\
&+& lim_{N \rightarrow \infty} N \mathbb{E} \bigg[ \frac{ \gamma N \frac{\Tilde{C}^N(t)^2}{2} \dot{f} (\zeta)}{1-F(\gamma N (\Tilde{B}^{N}(t))} \bigg] \\
\end{eqnarray*}

In the second term since $\Tilde{C}^N(t) = \frac{C^N(t)}{N}$ where $C^N(t) \sim \mathtt{Bin} (D^N(t), \frac{1}{N})$, we have

\begin{eqnarray*}
\lefteqn{\mathbb{E}(\Tilde{C}^N(t)^2) = \frac{1}{N^2} \mathbb{E} (C^N(t)^2)}\\
&=& \frac{1}{N^2} \bigg[ D^N(t) \frac{1}{N} (1-\frac{1}{N}) + (\frac{D^N(t)}{N})^2 \bigg]\\
&=& \frac{1}{N^2} \bigg[ \Tilde{D}^N(t) (1-\frac{1}{N}) + (\Tilde{D}^N(t))^2 \bigg]\\
\end{eqnarray*}

Also $\dot{f} (\zeta)$ is bounded, since we require the drift functions to satisfy the Lipschitz condition. Hence the second term vanishes as $N \to \infty$. In the first term, noting that $\gamma N \to \Gamma$ and $\mathbb{E} (\Tilde{C}^N(t)) = \frac{\Tilde{D}^N(t)}{N}$, given $\Tilde{B}^N(t)=b$, $\Tilde{D}^N(t) =d$, the limit becomes $\frac{\Gamma d f(\Gamma b)}{1-F(\Gamma b)}$.

\section{Fluid Limit for the HILT-SI Model}
\label{app:hilt-si-fluidlimit}
Adopting the scaling as explained in Section~\ref{sec:HILT-SI-fluid}, we can write the evolution of $(X^N(t), Y^{N}(t))$ as follows:

\begin{eqnarray*}
\lefteqn{X^N(t+1) - X^N(t)} \\
&=& \lambda \alpha_N \frac{(X^N(t) + Y^N(t))}{N} (A^N(t) - X^N(t))\\
&+& \frac{\gamma_N \frac{D^N(t)}{N}}{1 - \gamma_N B^N(t)} Y^N(t) + Z^N_x(t)\\
\end{eqnarray*}
\begin{eqnarray*}
\lefteqn{Y^N(t+1) - Y^N(t)} \\
&=& \lambda \sigma_N \frac{(X^N(t) + Y^N(t))}{N}(N- A^N(t) - Y^N(t))\\
&-& \frac{\gamma_N \frac{D^N(t)}{N}}{1 - \gamma_N B^N(t)} Y^N(t)+ Z^N_y(t)\\
\end{eqnarray*}
As earlier, defining $\Tilde{X}^N(t)$, $\Tilde{Y}^N(t)$, $\Tilde{A}^N(t)$, $\Tilde{B}^N(t)$, $\Tilde{D}^N(t)$ as the fractions of entire population, we can write the mean drift functions per unit update size for $\Tilde{X}^{N}(t)$ and $\Tilde{Y}^{N}(t)$ respectively as,
\[\frac{g^{N}_x(x,y,b,d)}{\frac{1}{N}} = \bigg( \lambda \alpha_N N (x+y) (a-x)  + \frac{\gamma_N N d}{1 - \gamma_N N b} y \bigg)\]
\[\frac{g^{N}_y(x,y,b,d)}{\frac{1}{N}} = \bigg( \lambda \sigma_N N (x+y) (1-a-y)  - \frac{\gamma_N N d}{1 - \gamma_N N b} y \bigg)\]

Let $g_x(x,y,b,d) =  \lambda \alpha (x+y) (a-x)  + \frac{\Gamma d}{1 - \Gamma b} y$ and $g_y(x,y,b,d) =  \lambda \sigma (x+y) (1-a-y)  - \frac{\Gamma d}{1 - \Gamma b} y$. Combining these with the limiting drift functions of the HILT model (see Appendix~\ref{app:hilt-fluidlimit}) we can derive the fluid limits of the HILT-SI model.

\section{Kurtz's Theorem and Proof of the Theorem~\ref{thm:kurtz-theorem-hilt}}
\label{app:kurtz}

Kurtz's theorem \cite{kurtz70ode-markov-jump-processes} provides us a way by which we can, subject to certain conditions, approximate the evolution of a pure jump Markov process by the solution of a derived ODE. In this section, we derive the necessary conditions for the convergence of the pure jump Markov process to the ODE solution, along the lines of Kurtz theorem, and subsequently check them for the HILT model.

For each $N \geq 1$, $Y^{(N)}(k), k \geq 0,$ is a DTMC on  $\Delta^{(N)} \subset \Delta$ where $\Delta$ is a subset of an appropriate dimensional Euclidean space. Let $e^{(N)}_j, 1 \leq j \leq J^{(N)},$ be the possible jumps of $Y^{(N)}(\cdot)$ and  $p^{(N)}_j(y), y \in \Delta^{(N)},$ the jump probabilities, i.e.,
\[ p^{(N)}_j(y) := P( Y^{(N)}(k+1) - Y^{(N)}(k) = e^{(N)}_j | Y^{(N)}(k) = y)\]
We define an indicator variable, $ I_j(Y^{(N)}(k)) = 1$ if the jump out of $Y^{(N)}(k)$ is $e^{(N)}_j$. Then we can write
\begin{eqnarray*}
  Y^{(N)}(k) &=& Y^{(N)}(0) + 
  \sum_{i=0}^{k-1} \sum_{j=1}^{J^{(N)}} e^{(N)}_j I_j(Y^{(N)}(i))   \\
   &=& Y^{(N)}(0) + \sum_{i=0}^{k-1} \sum_{j=1}^{J^{(N)}} e^{(N)}_j  p^{(N)}_j(Y^{(N)}(i))\\
  &+& \sum_{i=0}^{k-1} \sum_{j=1}^{J^{(N)}} e^{(N)}_j (I_j(Y^{(N)}(i)) -  p^{(N)}_j(Y^{(N)}(i))) \\ 
\end{eqnarray*}
where in the second equality, the second terms comprises the sum of the successive mean jumps, and the last term the sum of successive random "noise" terms. We rewrite these compactly by defining, for $i \geq 0$, 

Define,  for $i \geq 0$
\begin{eqnarray*}
 \xi^{(N)} (i) &:=& \sum_{j=1}^{J^{(N)}} e^{(N)}_j (I_j(Y^{(N)}(i)) -  p^{(N)}_j(Y^{(N)}(i))),
\end{eqnarray*}
so that 
\begin{eqnarray*}
  Y^{(N)}(k) &=& Y^{(N)}(0) +  \sum_{i=0}^{k-1} f^{(N)}(Y^{(N)}(i)) +  \sum_{i=0}^{k-1} \xi^{(N)}(i)   
\end{eqnarray*}
where $ f^{(N)}(y) := \sum_{j=1}^{J^{(N)}} e^{(N)}_j
p^{(N)}_j(y)$ is the expected drift out of the state $y \in \Delta^{(N)}$.
 
With $\mathcal{F}^{(N)}_k := \sigma(Y^{(N)}(0), Y^{(N)}(1), \cdots ,
Y^{(N)}(k))$, we easily see that
\begin{lemma}
  $Z^{(N)}(k) := \sum_{i=0}^{k-1} \xi^{(N)}(i)$, $k > 0$, is a martingale with respect to $\mathcal{F}^{(N)}_k$, $k \geq 0$.
\end{lemma}

For each $t \geq 0$, for each $N$, we can write
\begin{eqnarray}
\label{YN_floorNt_expansion}
\lefteqn{Y^{(N)}(\lfloor Nt \rfloor)} \\ 
&=& Y^{(N)}(0) + \sum_{i=0}^{\lfloor Nt \rfloor-1} f^{(N)}(Y^{(N)}(i)) +  Z^{(N)} (\lfloor Nt \rfloor)
   \hspace{1cm}
\end{eqnarray}
Suppose
\begin{eqnarray}
\label{condition_unif_convergence}
 \sup_{x \in \Delta^{(N)}}
  \|  \frac{f^{(N)}(x)}{\frac{1}{N}} -  f(x) \| \stackrel{N \to \infty}{\to} 0
\end{eqnarray}
with $f(x)$ being Lipschitz on $\Delta$, i.e., for all $x, y \in \Delta$, for some
$L > 0$,
\begin{eqnarray}
\label{condition_lipschitz}
   \| f(x) - f(y) \| \leq L \|x - y\|
\end{eqnarray}
Now, in \eqref{YN_floorNt_expansion}, if $Y^{(N)}(0) \stackrel{p}{\rightarrow} y(0)$, and the second term goes to
$0$, the process $Y^{(N)}(\lfloor Nt \rfloor), t \geq 0,$ appears to
be following
\begin{eqnarray*}
  y(t) = y(0) + \int_0^t f(y(u)) \mathrm{d} u
\end{eqnarray*}
Take $t \in [0, T]$. For given $t$, subtracting the ODE solution on both sides of \eqref{YN_floorNt_expansion}
\begin{eqnarray}
\label{eqn:combined_error_terms}
  \lefteqn{Y^{(N)}(\lfloor Nt \rfloor) - y(t)}\\ 
  &=& Y^{(N)}(0) - y(0) +  \sum_{i=0}^{\lfloor Nt \rfloor-1} f^{(N)}(Y^{(N)}(i))\\
  &-& \int_0^t f(y(u)) \mathrm{d} u + Z^{(N)} (\lfloor Nt \rfloor)
\end{eqnarray}
We will study each of the three "error" terms on the right hand side of this equation. 
First, we consider the term
\begin{eqnarray}
\lefteqn{\sum_{i=0}^{\lfloor Nt \rfloor-1} f^{(N)}(Y^{(N)}(i))-\int_0^t f(y(u))\mathrm{d} u }\\
&=& \sum_{i=0}^{\lfloor Nt \rfloor-1} f^{(N)}(Y^{(N)}(i))
 - \sum_{i=0}^{\lfloor Nt \rfloor-1} \frac{1}{N} f(Y^{(N)}(i)) \hspace{1cm}\label{term_1_error}\\
&& + \sum_{i=0}^{\lfloor Nt \rfloor-1} \frac{1}{N} f(Y^{(N)}(i)) -\sum_{i=0}^{\lfloor Nt \rfloor-1} \frac{1}{N} f \bigg( y \bigg( \frac{i}{N} \bigg) \bigg)\hspace{0.5cm}\label{term_3_error}\\
&& + \sum_{i=0}^{\lfloor Nt \rfloor-1} \frac{1}{N} f \bigg( y \bigg( \frac{i}{N} \bigg) \bigg)
  - \int_0^t f(y(u)) \mathrm{d} u \hspace{1cm} \label{term_2_error}
\end{eqnarray}

\paragraph*{Term \eqref{term_1_error}}
Using the uniform convergence hypothesis, for each $t$, and for an $\epsilon_1 >0$ there exists an $N$ large so that,
\begin{eqnarray*}
  \lefteqn{  \left\|\sum_{i=0}^{\lfloor Nt \rfloor-1} f^{(N)}(Y^{(N)}(i))
    - \sum_{i=0}^{\lfloor Nt \rfloor-1} \frac{1}{N} f(Y^{(N)}(i))\right\|} \\
  &\leq& \sum_{i=0}^{\lfloor Nt \rfloor-1} \frac{1}{N}
  \left\| \frac{f^{(N)}(Y^{(N)}(i))}{\frac{1}{N}} -  f(Y^{(N)}(i))\right\|  \\
  &\leq& \sum_{i=0}^{\lfloor Nt \rfloor-1} \frac{1}{N} \epsilon_1 = \frac{\lfloor Nt \rfloor}{N} \epsilon_1
\end{eqnarray*}

\paragraph*{Term \eqref{term_2_error}}
By the hypothesis in \eqref{condition_unif_convergence}, $f(\cdot)$ is Lipschitz on $\Delta$. The following can then be proved:
\begin{enumerate}
\item $||f(y)||\leq ||f(0)|| + L||y||$ (i.e., at most linear growth in $||y||$)
\item The ODE $\dot{y}(t) = f(y(t)),$ with initial value $y(0)$, has
  the unique solution $ y(t) = y(0) + \int_0^t f(y(u)) \mathrm{d} u$. Further, over $0 \leq t \leq T$, there are constants such that
  \begin{enumerate}
  \item $||f(y(t)|| \leq K_1(T) (1 + ||y(0)||)$
  \item $||y(t + \tau) - y(t)|| \leq K_1(T) (1 + ||y(0)||) \tau$
  \end{enumerate}
\end{enumerate}

\begin{eqnarray*}
\lefteqn{  \left\|\sum_{i=0}^{\lfloor Nt \rfloor-1} \frac{1}{N} f \bigg( y \bigg( \frac{i}{N} \bigg) \bigg)
  - \int_0^t f(y(u)) \mathrm{d} u \right\| }\\
&\leq & \sum_{i=0}^{\lfloor Nt \rfloor-1} \left\|\frac{1}{N} f \bigg( y \bigg( \frac{i}{N} \bigg) \bigg)
     -  \int_{\frac{i}{N}}^{\frac{i+1}{N}} f(y(u)) \mathrm{d} u\right\|\\
&+& \left\|\int_{\frac{\lfloor Nt \rfloor}{N}}^{t} f(y(u)) \mathrm{d} u \right\| \\
&\leq & \sum_{i=0}^{\lfloor Nt \rfloor-1} \frac{1}{2N} \left\| \bigg[ f \bigg( y \bigg( \frac{i+1}{N} \bigg) \bigg) -  f \bigg( y \bigg( \frac{i}{N} \bigg) \bigg)\right\| 
                                               + c'_1(T) \frac{1}{N}\\
&\leq & \sum_{i=0}^{\lfloor Nt \rfloor-1} c_1(T) \frac{1}{2 N^2} 
                                               + c'_1(T) \frac{1}{N}\\
&=& c_1(T) \frac{\lfloor Nt \rfloor}{2 N^2} + c'_1(T) \frac{1}{N}\\
\end{eqnarray*}
where the penultimate inequality follows from the Lipschitz property of $f(\cdot)$.

\paragraph*{Term \eqref{term_3_error}}
By using the Lipschitz property of $f(\cdot)$
\begin{eqnarray*}
  \lefteqn{ \left\| \sum_{i=0}^{\lfloor Nt \rfloor-1} \frac{1}{N} f(Y^{(N)}(i)) -
\sum_{i=0}^{\lfloor Nt \rfloor-1} \frac{1}{N} f \bigg( y \bigg( \frac{i}{N} \bigg) \bigg) \right\|} \\
&\leq& \sum_{i=0}^{\lfloor Nt \rfloor-1} \frac{1}{N} L \left\|(Y^{(N)}(i)) -
y \bigg( \frac{i}{N} \bigg) \right\|
\end{eqnarray*}

We now work on the left hand side of \eqref{eqn:combined_error_terms}
\begin{eqnarray}
\label{eqn:error_term_expanded}
  \lefteqn{  \left\| Y^{(N)}(\lfloor Nt \rfloor) - y(t) \right\| } \\ \nonumber
  &=& \left\|Y^{(N)}(\lfloor Nt \rfloor) - y\left(\frac{\lfloor Nt \rfloor}{N}\right)
  - \left( y(t) -  y\left(\frac{\lfloor Nt \rfloor}{N}\right)\right) \right\| \\ \nonumber
  &\geq& \left\|Y^{(N)}(\lfloor Nt \rfloor)
  - y\left(\frac{\lfloor Nt \rfloor}{N}\right) \right\|
  - \left\| \left( y(t) -  y\left(\frac{\lfloor Nt \rfloor}{N}\right)\right) \right\|
\end{eqnarray}
and it can be shown that, with $y(t) = y(0) + \int_0^t f(y(u))
\mathrm{d} u$, for $t \in [0,T]$,
\begin{eqnarray}
\label{eqn:using_lipschitz}
  \left\| ( y(t) 
 -  y\left(\frac{\lfloor Nt \rfloor}{N}\right) \right\| \leq c_4(T) \frac{1}{N}
\end{eqnarray}

We now use the hypothesis that $Y^{(N)}(0) \stackrel{p}{\to} y(0)$

Then, for large enough $N$, combining \eqref{term_1_error}, \eqref{term_2_error}, \eqref{term_3_error}, \eqref{eqn:error_term_expanded}, \eqref{eqn:using_lipschitz} with \eqref{eqn:combined_error_terms} for all $t \in [0,T]$, we can write
\begin{eqnarray}
\label{eqn:post_processed_error}
 \lefteqn{\left\|Y^{(N)}(\lfloor Nt \rfloor) - y(\frac{\lfloor Nt \rfloor}{N})\right\| }\\ \nonumber
  &\leq& \epsilon +  \sum_{i=0}^{\lfloor Nt \rfloor-1} \frac{1}{N} L \left\|(Y^{(N)}(i)) -
  y(\frac{i}{N}) \right\| \\ \nonumber
  && + \left\| Z^{(N)} (\lfloor Nt \rfloor) \right\| 
\end{eqnarray}

We now work on the "noise" martingale, using Burkholder Inequality(BI) for zero-mean martingales, 
\begin{eqnarray*}
\lefteqn{\mathsf{E}\left(\sup_{0 \leq t \leq T} 
  \left\|Z^{(N)} (\lfloor Nt \rfloor) \right\|^2 \right) }\\
&=& \mathsf{E} \left(
\sup_{0 \leq t \leq T} \left\|\sum_{i=0}^{\lfloor Nt \rfloor -1} \xi^{(N)}(i)\right\|^2 \right)\\
&\stackrel{BI}{\leq}& c_1 \mathsf{E} 
     \left(\sum_{i=0}^{\lfloor NT \rfloor -1} \left\|\xi^{(N)}(i)\right\|^2 \right)\\
& \leq & c_1 c_0 \frac{\lfloor NT \rfloor}{N^2}
\end{eqnarray*}
Where we have \emph{required} that $\mathsf{E} \left\| \xi^{(N)} (i)\right\|^2
\leq c_0\frac{1}{N^2}$, for all $i$.

It follows that 
\begin{eqnarray*}
  P\left(\sup_{0 \leq t \leq T} \left\| Z^{(N)} (\lfloor Nt \rfloor)\right\| 
  \leq \epsilon\right) \stackrel{N \to \infty}{\to} 1
\end{eqnarray*}
%

\begin{theorem} (Gronwall's Inequality)
  Given a positive sequence $a_n , n \geq0,$ and $b, c \geq 0$, if, $x_n,
  n \geq 0,$ is a nonnegative sequence such that, for each $k \geq 1$
\begin{eqnarray*}
  x_k \leq b + c \sum_{i=0}^{k-1} a_i x_i
\end{eqnarray*}
then, for each $k \geq 1$
\begin{eqnarray*}
  x_k \leq b e^{c \sum_{i=0}^{k-1} a_i }
\end{eqnarray*}
\end{theorem}

From \ref{eqn:post_processed_error}, on a sample path where $\sup_{0 \leq t \leq T}\| Z^{(N)}(\lfloor Nt
\rfloor)\| \leq \epsilon$ and $|| Y^{N}(0) - y(0)|| < epsilon$, we can write, for $0 \leq k \leq \lfloor NT
\rfloor$,
\begin{eqnarray*}
  \left\|Y^{(N)}(k) - y(\frac{k}{N})\right\| 
  &\leq& 3 \epsilon +  \sum_{i=0}^{k-1} \frac{1}{N} L \left\|(Y^{(N)}(i)) -
  y(\frac{i}{N}) \right\|
\end{eqnarray*}
Hence, by Gronwall's inequality, for any $\epsilon > 0$,
\begin{eqnarray*}
\lefteqn{ \mathbb{P}(\sup_{0 \leq t \leq T}\|Z^{(N)}(\lfloor Nt \rfloor)\|\leq \epsilon)}\\
  &\leq& \mathbb{P}\left(\sup_{0 \leq k \leq \lfloor NT \rfloor}
  \left\|Y^{(N)}(k) - y(\frac{k}{N})\right\| 
               \leq 2\epsilon e^{\frac{L  \lfloor NT \rfloor}{N}} \right)
\end{eqnarray*}
By the earlier assertion, the left hand term $\to 1$ as $N \to \infty$,  hence
so does the right hand term. Hence we can state the following theorem. 

\begin{theorem}
\label{thm:kurtz-equivalent}
With the notation defined earlier, if,  for all $k \geq 0$,
\begin{eqnarray*}
  \mathsf{E} \left( 
      \left\|\sum_{j=1}^{J^{(N)}} e^{(N)}_j (I_j(Y^{(N)}(k))) - f^{(N)}(Y^{(N)}(k))\right\|^2 
    |  \mathcal{F}_{k} \right) \leq c_0 \frac{1}{N^2}
\end{eqnarray*}
Further, $ \sup_{x \in \Delta^{(N)}} \left\|
  \frac{f^{(N)}(x)}{\frac{1}{N}} - f(x) \right\| \stackrel{n \to
  \infty}{\to} 0$ with $f(x)$ being Lipschitz (i.e., for all $x, y \in
\Delta$, for some $L > 0$, $ \left\| f(x) - f(y) \right\| \leq L
\left\|x - y\right\|$).  Define $y(t), 0 \leq t \leq T,$ by $ y(t) =
y(0) + \int_0^t f(y(u)) \mathrm{d} u $. Also, $Y^{(N)}(0)
\stackrel{p}{\to} y(0)$.

Then, for each $T > 0$, and each $\epsilon > 0$,
\begin{eqnarray*}
  P\left(\sup_{0 \leq t \leq T}
    \left\|Y^{(N)}( \lfloor Nt \rfloor) - y(t)\right\| 
    > \epsilon \right)  \stackrel{N \to \infty}{\to} 0
\end{eqnarray*}
\end{theorem}

\begin{flushright}
$\blacksquare$
\end{flushright}

\section{Application of Kurtz's theorem}
\label{app:kurtz-applied}
\subsection{Verification of Kurtz's theorem for HILT model}
\label{app:kurtz-applied-hilt}
Now we can check each of the above conditions in Theorem \ref{thm:kurtz-equivalent} for the HILT model.

\begin{itemize}
\item[(i)]\textbf{Lipschitz property}

Consider,
\[g_1(b,d)= d\]
\[ g_2(b,d) = h_F(\Gamma b) \Gamma d (1 - b - d) - d\]

\[ \frac{\partial g_1}{\partial b} = 0 ;  \frac{\partial f_1}{\partial d} =  1\]

\[ \frac{\partial g_2}{\partial b} = \Gamma^2 d \dot{h}_F (\Gamma b) (1-b-d) - \Gamma d h_F(\Gamma b); \frac{\partial g_2}{\partial d} = h_f(\Gamma b) \Gamma (1-b-2d) -1 \]

Under the assumption that $\dot{f}(\cdot)$ is bounded, we see that each of the terms above is bounded when $(b,d) \in [0,1]\times[0,1-b]$. Thus the norm of Jacobian $|| Dg(b,d) ||$ is uniformly bounded, and it follows that $g(b,d) = (g_1(b,d),g_2(b,d))$ is Lipschitz.

\item[(ii)]\textbf{Uniform Convergence}

We know that

\begin{eqnarray*}
g^{N}(\Tilde{B}^{N}(k),\Tilde{D}^{N}(k)) := \frac{1}{N} \bigg(\Tilde{D}^{N}(k),  N \mathbb{E} \bigg[ \frac{F(\gamma_N N (\Tilde{B}^{N}(t) + \Tilde{C}^{N}(t)))- F(\gamma_N N (\Tilde{B}^{N}(t))}{1-F(\gamma_N N (\Tilde{B}^{N}(t))} \bigg] \times \bigg( 1-\Tilde{B}^{N}(t)- \Tilde{D}^{N}(t) \bigg) - \Tilde{D}^{N}(t) \bigg) 
\end{eqnarray*}

Let $g_1(b,d)= d$ and $g_2(b,d) =  \frac{\mathbf{f}(\Gamma b)\Gamma d}{1- F(\Gamma b)} (1-b-d) - d$ and define,

\[ g(b,d) := \bigg( g_1(b,d), g_2(b,d) \bigg) \]

By definition, $\lim_{N \rightarrow \infty} N \gamma =\Gamma$ and by the steps in Appendix~\ref{app:limit_expectation}, the uniform convergence of $\frac{g^{N}(b,d)}{1/N}$ to $g(b,d)$ in the domain $(b,d) \in [0,\frac{1}{N},\cdots ,N]\times[0,\frac{1}{N}, \cdots,b]$ is proven. 

\item[(iii)]\textbf{Bounded Noise variance}
We will write the noise variance terms for the scaled processes $B^N(k)$ and $D^N(k)$, and from them derive the noise variance terms for the density dependent processes by scaling them by $\frac{1}{N^2}$. 

Let $Z^{N}_b(k)$ and $Z^{N}_d(k)$ be the zero mean random variables representing the noise in the drift terms of $B^N(k)$ and $D^N(k)$ respectively. We can then write,

\[ E[ | Z^{N}_b(k)|^2 | \mathcal{F}_k] = (1-\frac{1}{N})\frac{1}{N} D^N(k) \]
\[ E[ | Z^{N}_d(k)|^2 | \mathcal{F}_k] = (1-\frac{1}{N})\frac{1}{N} D^N(k) + \sum_{z=0}^{z=D^N(k)} g(z) (1-g(z)) (N- B^N(k) - D^N(k)) {D^N(k) \choose z} \frac{1}{N}^z (1-\frac{1}{N})^{D^N(k)-z} \]

where, 
\[ g(z) = \frac{F(\gamma N (B^N(k)+z)) - F(\gamma N B^N(k))}{1 - F(\gamma N B^N(k))} \]

Note that both these terms can be upper bounded by constants say $c_b$ and $c_d$ and hence the noise conditions for $\Tilde{B}^N(k)$ and $\Tilde{D}^N(k)$ are satisfied.

\item[(iv)]\textbf{Convergence of initial conditions}
By choice, we have $\Tilde{B}^{N}(0) = b(0)$ and $\Tilde{D}^{N}(0)=d(0)$. 

\end{itemize}
Thus by Theorem \ref{thm:kurtz-equivalent}, we have for each $T > 0$ and each $\epsilon > 0$,

\[ P \bigg( \sup_{0 \leq t \leq T} \big| \big| \big( \Tilde{B}^{N}( \lfloor Nt \rfloor ),\Tilde{D}^{N}( \lfloor Nt \rfloor ) \big) - \big( b(t),d(t) \big) \big| \big| > \epsilon \bigg)  \stackrel{N\rightarrow \infty}{\rightarrow} 0\]
where $(b(t),d(t))$ is the unique solution of the ODE

\[\dot{b}(t) = d(t) \]
\[\dot{d(t)} = -d(t) + \frac{\Gamma d(t)}{1 - \Gamma b(t)} (1 - b(t) - d(t))\]

with initial conditions $(b(0)=0,d(0)=a(0))$.

\begin{flushright}
$\blacksquare$
\end{flushright}

\subsection{Verification of Kurtz theorem for HILT-SI model}
\label{app:kurtz-applied-hilt-si}

The first two equations, as proved above will remain unaltered. Consider the drift equations for $\Tilde{X}^N(k)$ and $\Tilde{Y}^N(k)$ and the corresponding ODEs. We can use a similar procedure as above to verify the conditions for applying Kurtz's theorem.

\begin{itemize}
\item[(i)]\textbf{Lipschitz property}

Consider,
\[g_x(x,y,b,d)= \lambda \alpha (x+y) (a-x)  + \frac{\Gamma d}{1 - \Gamma b} y\]
\[ g_y(x,y,b,d) = \lambda \sigma (x+y) (1-a-y)  - \frac{\Gamma d}{1 - \Gamma b} y \]

where $a=b+d$. We can calculate the partial derivatives of $g_x$ and $g_y$ with respect to $x$, $y$, $b$, $d$ and note that given $b+d \leq 1$, $x \leq b+d$,  $y \leq 1-b-d$ and all quantities non-negative, we have that the Jacobian is uniformly bounded and hence the drift functions are Lipschitz. 

\item[(ii)]\textbf{Uniform Convergence}
We know that,
\[g^{N}_x(x,y,b,d) = \frac{1}{N} \bigg( \lambda \alpha_N N (x+y) (a-x)  + \frac{\gamma_N N d}{1 - \gamma N b} y \bigg)\]
\[g^{N}_y(x,y,b,d) =  \frac{1}{N} \bigg( \lambda \sigma_N N (x+y) (1-a-y)  - \frac{\gamma_N N d}{1 - \gamma N b} y \bigg)\]

Define,
\[g_x(x,y,b,d) = \lambda \alpha (x+y) (a-x)  + \frac{\Gamma d}{1 - \Gamma b} y\]
\[g_y(x,y,b,d) =  \lambda \sigma (x+y) (1-a-y)  - \frac{\Gamma d}{1 -\Gamma b} y \]

Since as $N \rightarrow \infty$, we have $\gamma_N N \rightarrow \Gamma$, $\alpha_N N \rightarrow \alpha$, $\sigma_N N \rightarrow \sigma$ and we see that the uniform convergence is straightforward.

\item[(iii)]\textbf{Bounded Noise variance}
As earlier we will write the noise variance terms for $X^N(k)$ and $Y^N(k)$ and derive the density dependent process' noise variance by scaling by $\frac{1}{N^2}$. 

Let $Z^{N}_x(k)$ and $Z^{N}_y(k)$ be the zero mean random variables representing the noise in the drift terms of $X^N(k)$ and $Y^N(k)$ respectively. We can then write,

\begin{eqnarray*}
E[ | Z^{N}_x(k)|^2 | \mathcal{F}_k] &=& \sum_{z_1=0}^{D^N(k)} \frac{\gamma z_1}{1-\gamma B^N(k)} \bigg(1- \frac{\gamma z_1}{1-\gamma B^N(k)}\bigg) Y^N(k) {D^N(k) \choose z_1} \frac{1}{N}^z_1 (1-\frac{1}{N})^{D^N(k)-z_1} \\
&+& \sum_{z_2=0}^{X^N(k)+Y^N(k)} \frac{\lambda \alpha z_2}{N^2} (1- \frac{\lambda \alpha z_2}{N^2}) (A_k - X_k) {X^N(k) + Y^N(k) \choose z_2} \frac{1}{N}^z_2 (1-\frac{1}{N})^{X^N(k)+Y^N(k) - z_2}\\
\end{eqnarray*}

\begin{eqnarray*}
E[ | Z^{N}_y(k)|^2 | \mathcal{F}_k] &=& \sum_{z_1=0}^{D^N(k)} \frac{\gamma z_1}{1-\gamma B^N(k)} \bigg(1- \frac{\gamma z_1}{1-\gamma B^N(k)}\bigg) Y^N(k) {D^N(k) \choose z_1} \frac{1}{N}^z_1 (1-\frac{1}{N})^{D^N(k)-z_1} \\
&+& \sum_{z_2=0}^{X^N(k)+Y^N(k)} \frac{\lambda \sigma z_2}{N^2} (1- \frac{\lambda \sigma z_2}{N^2}) (N- A_k - Y_k) {X^N(k) + Y^N(k) \choose z_2} \frac{1}{N}^z_2 (1-\frac{1}{N})^{X^N(k)+Y^N(k) - z_2}\\
\end{eqnarray*}

We see that the above two terms can be bound by constants $c_x$ and $c_y$ respectively and we can see that the noise conditions are satisfied.

\item[(iv)]\textbf{Convergence of initial conditions}
By choice, we have $\Tilde{X}^{N}(0) = x(0)$ and $\Tilde{Y}^{N}(0)=y(0)$. 

\end{itemize}

Thus by Theorem \ref{thm:kurtz-equivalent}, we have for each $T > 0$ and each $\epsilon > 0$,

\[ P \bigg( \sup_{0 \leq t \leq T} \big| \big| \big( \Tilde{B}^{N}( \lfloor Nt \rfloor ),\Tilde{D}^{N}( \lfloor Nt \rfloor ), \Tilde{X}^{N}( \lfloor Nt \rfloor ),\Tilde{Y}^{N}( \lfloor Nt \rfloor ) \big) - \big( b(t),d(t), x(t), y(t) \big) \big| \big| > \epsilon \bigg)  \stackrel{N\rightarrow \infty}{\rightarrow} 0\]
where $(b(t),d(t), x(t), y(t))$ is the unique solution of the ODE

\[ \dot{b} = d\]
\[ \dot{d} = \frac{\Gamma d}{1 - \Gamma b} (1-b-d) - d\]
\[ \dot{x} = \lambda \alpha (x+y) (a-x)  + \frac{\Gamma d}{1 - \Gamma b} y \]
\[ \dot{y} = \lambda \sigma (x+y) (1-a-y)  - \frac{\Gamma d}{1 - \Gamma b} y \]

with initial conditions $(b(0)=0,d(0)=d_0, x(0)=x_0 , y(0)=0)$.

\begin{flushright}
$\blacksquare$
\end{flushright}

\section{Solution of the ODE}
\label{app:ode-solving}
Given the system of ODEs 
\[ \dot{b} = d\]
\[ \dot{d} = -d + \frac{\Gamma d}{1 - \Gamma b} (1 -b - d) \]

Dividing $\dot{d}$ by $\dot{b}$, 

\[ \frac{\mathrm{d}d}{\mathrm{d}b} = \frac{\Gamma \dot{b} -\dot{b} -\Gamma \dot{b}^2}{d(1-\Gamma b)}\]

By separating the variables,

\[ \frac{\mathrm{d}d}{\Gamma -1-\Gamma d} = \frac{\mathrm{d}b}{1-\Gamma b} \]
Integrating on both sides, and after taking anti-logarithm

\[ \Gamma - 1 - \Gamma d = c_1 (1-\Gamma b) \]

Differentiating on both sides yields $\dot{d} = c_1 d$ and hence $d(t) = c_2 e^{c_1 t}$. Substituting in the above equation for $d(t)$ we get 

\[b(t) = \frac{c_2}{c_1} e^{c_1 t} + \frac{1+ c_1 - \Gamma}{\Gamma c_1} \]

Solving for constants using the initial conditions $b(0)=0$, $d(0)=d_0$ results in $ c1= -(1 + \Gamma d_0 -\Gamma) =: -r$ and $c_2=d_0$.

Hence we have,

\[b(t) = \frac{d_0}{r} - \frac{d_0}{r} e^{-rt} \]
\[d(t) = d_0 e^{-rt} \]
\begin{flushright}
$\blacksquare$
\end{flushright}

\section{Convergence of $h_{\gamma}^{(N)}(k)$ to $b_\infty$}
\label{app:hilt-convergence}
The solution of the ODE suggests that $\lim_{t \to \infty} b(t) = d_0/r$. This is consistent with the fact that $\lim_{t \to \infty} \frac{h_{\gamma}^{(N)}(k)}{N} \rightarrow \frac{d_0}{1 - (1-d_0)\Gamma} = b_\infty$ as we now proceed to show.

\[ h_{\gamma}^{(N)}(k) = k \bigg[ 1 + (N-k)\gamma \bigg[ 1 + (N-k-1)\gamma \bigg[ 1 + \cdots \]

\[ h_{\gamma}^{(N)}(k+1) = (k+1) \bigg[ 1 + (N-k-1)\gamma \bigg[ 1 + \cdots \]

Thus we can write,

\[ h_{\gamma}^{(N)}(k) = k \bigg[ 1 + \gamma (N-k) \frac{h_{\gamma}^{(N)}(k+1)}{k+1} \bigg] \]

Now substituting for $k=d_0 N$ and noting that $\Gamma = \gamma N$ we have

\[  \frac{h_{\gamma}^{(N)}(N d_0)}{N d_0} = 1 + \Gamma (1-d_0) \frac{h_{\gamma}^{(N)}(N(d_0 + \frac{1}{N}))}{N(d_0 + \frac{1}{N})} \]

Taking $N \rightarrow \infty$ and noting that $h_{\gamma}^{(N)}(k)$ is a continuous function, we have

\[  \frac{1}{d_0} \frac{h_{\gamma}^{(N)}(N d_0)}{N} = 1 + \frac{\Gamma (1-d_0)}{d_0} \frac{h_{\gamma}^{(N)}(N d_0)}{N} \]

Take limits on both sides, and solving for the unknown, 
\[\frac{h_{\gamma}^{(N)}(N d_0)}{N} \rightarrow \frac{d_0}{1 - (1-d_0)\Gamma} = b_\infty\]

\begin{flushright}
$\blacksquare$
\end{flushright}

\end{document}